\documentclass{IEEEtran}

\usepackage{tikz}
\usetikzlibrary {positioning}
\usepackage{pgf}

\usepackage{cite}
\usepackage{graphicx}
\usepackage{epstopdf}
\usepackage{amsmath}
\usepackage{amsthm}
\usepackage{amsmath,amssymb,amsfonts,amsthm}
\usepackage{verbatim}
\usepackage{url}
\usepackage{comment}
\usepackage{graphicx}
\usepackage{color}
\usepackage{xcolor}

\usepackage{qcircuit}
\usepackage{diagbox}
\usepackage{makecell}
\usepackage{braket}
\usepackage{mathtools}
\usepackage{float}
\usepackage{subfig}
\usepackage{multicol}
\usepackage{listings}
\newtheorem{theorem}{Theorem}
\newtheorem{lemma}{Lemma}

\newcommand{\ncost}[1]{\ensuremath{\cos(\frac{\theta}{#1})}}
\newcommand{\nsint}[1]{\ensuremath{\sin(\frac{\theta}{#1})}}

\newcommand{\ncosa}[1]{\ensuremath{\cos(\frac{\alpha}{#1})}}
\newcommand{\nsina}[1]{\ensuremath{\sin(\frac{\alpha}{#1})}}

\newcommand{\ck}{\mathcal{C}}
\newcommand{\ckt}{\mathcal{\widehat{C}}}
\newcommand{\M}{\mathcal{M}}

\DeclarePairedDelimiter{\abs}{\lvert}{\rvert}

\title{On Actual Preparation of Dicke State on a Quantum Computer}

\author{\IEEEauthorblockN{
Chandra Sekhar Mukherjee~\IEEEauthorrefmark{1}\IEEEauthorrefmark{3},
Subhamoy Maitra~\IEEEauthorrefmark{1}\IEEEauthorrefmark{4}, 
Vineet Gaurav~\IEEEauthorrefmark{2}\IEEEauthorrefmark{5} and
Dibyendu Roy~\IEEEauthorrefmark{1}\IEEEauthorrefmark{6} \\}
\IEEEauthorblockA{
\IEEEauthorrefmark{1} Indian Statistical Institute, Kolkata,
\IEEEauthorrefmark{2} Indian Institute of Science Education and Research, Mohali \\
\thanks{ 
~\IEEEauthorrefmark{3} chandrasekhar.mukherjee07@gmail.com
~\IEEEauthorrefmark{4} subho@isical.ac.in
~\IEEEauthorrefmark{5} vineet.gaurav1@gmail.com
~\IEEEauthorrefmark{6} roydibyendu.rd@gmail.com}
}}

\begin{document}

\maketitle

\begin{abstract}
The exact number of CNOT and single qubit gates needed to implement 
a Quantum Algorithm in a given architecture is one of the central problems of Quantum Computation. 
In this work we study the importance of concise realizations of Partially defined
Unitary Transformations for better circuit construction using the case study of Dicke State Preparation.
The Dicke States $(\ket{D^n_k})$ are an important class of entangled states with uses in 
many branches of Quantum Information.
In this regard we provide the most efficient Deterministic Dicke State Preparation Circuit
in terms of CNOT and single qubit gate counts in comparison to existing literature.
We further observe that our improvements also reduce architectural constraints of the circuits. 
We implement the circuit for preparing $\ket{D^4_2}$ on the ``ibmqx2'' machine 
of the IBM QX service and observe that the error induced due 
to noise in the system is lesser in comparison to the existing circuit descriptions.
We conclude by describing the CNOT map of the generic $\ket{D^n_k}$ preparation circuit and analyze
different ways of distributing the CNOT gates in the circuit and its affect on the 
induced error.
\end{abstract}




\section{Introduction}
\label{sec:1}
One of the most fundamental aspects of Quantum Mechanics is Quantum Computation.
Quantum Computers enable Quantum Algorithms that can perform operations with
even super exponential speed-ups in time over the best known classical algorithms.
Any quantum algorithm can be defined as a series of unitary transformations
and can be implemented as a Quantum Circuit.
A quantum circuit has a discrete set of gates such that their
combinations can express any unitary transformation with any desired accuracy.
Such a set of gates is called a universal set of gates. 
We know from the fundamental work by Barenco et.al~\cite{barenco} 
that single qubit gates and the controlled NOT (CNOT) gate form a 
universal set of gates. We call these gates as elementary gates.

Quantum State Preparation is a topic within Quantum Computation that has garnered 
interest in the past two decades due to applications of special quantum states in
several fields of Quantum Information Theory. 
A $n$-qubit quantum state $\ket{\psi_n}$ can be expressed as the superposition of $2^n$ orthonormal
basis states. In this work we look at $n$ qubit states as super position of the computational basis
states
$\ket{x_1x_2\ldots x_n}, x_i \in\{0,1\},~ 1 \leq i \leq n$.
The basis states in the expression of $\ket{\psi_n}$ with non zero amplitude are
called the active basis states.
Starting from the state $\ket{0}^{\otimes n}$ any arbitrary quantum state
can be formed using $\mathcal{O}(2^n)$ elementary gates, although 
for many $n$ qubit states preparation circuits with polynomial (in $n$)
number of elementary gates is possible. 
The family of Dicke States  $\ket{D^n_k}$ is one such example.
$\ket{D_k^n}$ is the $n$-qubit state 
which is the equal superposition state of all $n \choose w$
basis states of weight $k$. For example 
$\ket{D^3_1}=\frac{1}{\sqrt{3}}(\ket{001}+\ket{010}+\ket{100})$.
Dicke states are an interesting family of states due to the fact that 
they have $n \choose k$ active basis states, which can be exponential in $n$
when $k=\mathcal{O}(n)$ but need only polynomial number of elementary gates 
to prepare.
Dicke states also have applications in the areas of Quantum Game Theory, Quantum 
Networking, among others. 
One can refer to~\cite{dicke} for getting a more in-depth view of these applications. 

There has been several probabilistic and deterministic Dicke state algorithms 
designed in the last two decades \cite{dicke1,dicke2,dicke3}.
In this paper we focus on the algorithm described by B{\"a}rtschi et.al \cite{dicke} 
which gives a deterministic algorithm that takes $\mathcal{O}(kn)$ CNOT gates and 
$\mathcal{O}(n)$ depth to prepare the state $\ket{D^n_k}$. 
To the best of our knowledge this circuit description has the best gate count among the 
deterministic algorithms.
Here it is important to note that the paper by Cruz et.al \cite{dn1} describes two algorithms for
preparing the $\ket{D^n_1}$ states, also known as $W_n$ states. Both the algorithms have better gate
count than the  description by B{\"a}rtschi et.al \cite{dicke} and one of the algorithms has logarithmic
depth. However, their work is restricted to $\ket{D^n_1}$ and has no implication on the circuits for
$\ket{D^n_k},~2 \leq k \leq n-2$. We further observe in Section~\ref{sec:4} that the circuit 
obtained by us after the improvements for $\ket{D^n_1}$ is same as the linear $W_n$ circuit described in \cite{dn1}.

Because of the noisy behavior of current generation Quantum Computers the exact number of 
elementary gates needed and the distribution of the gates over the corresponding circuit become crucial issues
which need to be optimized in order to prepare a state with high fidelity. 
An example of a very recent work done in this area is \cite{aes} which reduces the gate count of AES implementation.
In this regard we discuss the following important problems in the domain of Quantum Circuit Design.

A unitary transformation acting on $n$ qubits can be expressed as a 
$2^n \times 2^n$ unitary matrix and can be decomposed into 
elementary gates in several ways.
Therefore finding the decomposition that needs the least amount of elementary gates 
is a very fundamental problem, with~\cite{song},~\cite{work} being examples of work done in this area. 
It is crucial to minimize the number of gates while decomposing a unitary matrix 
as every gate induces some amount of error into the result.
Especially reducing the number of CNOT gates is of importance due to the well known fact that
it induces more error compared to single qubit gates.

In this work we first describe a fundamental problem that 
decomposition of matrix using a universal set of gates poses. 
Let there be a unitary transformation that is to be performed on
a system of $n$ qubits. This task can be represented as a unitary matrix $U_n$
that works on the Hilbert Space $H_n$ of dimension $2^n$.
If we know the intended transformation for all the states of any 
orthonormal basis of $H_n$, that completely defines the unitary matrix $U_n$.
Let us consider such a transformation for $n=1$.
If the transformation is defined for the two states in the
computational basis $\ket{0}$ and $\ket{1}$ then the 
corresponding unitary matrix is completely defined.
If the transformation is defined as $\ket{0} \rightarrow \frac{1}{\sqrt{2}}(\ket{0}+\ket{1})$
and $\ket{1} \rightarrow \frac{1}{\sqrt{2}}(\ket{0}-\ket{1})$
then the corresponding matrix is the Hadamard matrix, expressed as 
$\begin{bmatrix}
\frac{1}{\sqrt{2}} & \frac{1}{\sqrt{2}}\\
\frac{1}{\sqrt{2}} & -\frac{1}{\sqrt{2}}
\end{bmatrix}$.
However if the transformation is only defined for one state,
$\ket{0} \rightarrow \frac{1}{\sqrt{2}}(\ket{0}+\ket{1})$
and not defined for $\ket{1}$ then there can be uncountably many 
unitary matrices that can perform the said transformation.
Specifically, any matrix of the form
$\begin{bmatrix}
\frac{1}{\sqrt{2}} & \alpha\\
\frac{1}{\sqrt{2}} & -\alpha
\end{bmatrix}$ 
can perform this task, where $\alpha \in \mathbb{C},~ \abs{\alpha}^2=\frac{1}{2}$.

There exists many quantum algorithms where at a step a particular 
transformation on $n$ qubits is defined only for a  
a subset of the states of a orthonormal basis.
This creates the possibility of there being uncountably many 
unitary matrices capable of such a transformation.
The algorithm described in~\cite{dicke} contains such transformations
that are not completely defined for all basis states. 
We call such a transformation a partially defined unitary transformation on $n$ qubits.
There are possibly multiple unitary matrices that can perform this transformation.
In that case it becomes an important problem to find out which candidate unitary matrix
can be decomposed using the minimal number of elementary gates.
 
Furthermore, the number of elementary gates needed to implement a well defined
Quantum Circuit also varies with the architecture of the actual Quantum Computer. 
The architectures of current generation Quantum Computers do not allow for CNOT gates to be 
implemented between any two arbitrary qubits. This CNOT constraint may further increase
the total number of CNOT and single qubit gates needed to implement a Quantum Circuit
on a specific Quantum Architecture. 
Against this backdrop, let us draw out the organization of the 
rest of the paper along with our contributions.

\subsection{Organization and Contribution}
In Section~\ref{sec:2} we first describe the preliminaries
needed to support our work. We first define the concept of 
maximally partial unitary transformation. We then describe the 
the circuit in~\cite{dicke} for preparing Dicke States.
We denote the circuit described in~\cite{dicke} for preparing $\ket{D^n_k}$
as $\ck_{n,k}$.

We start Section~\ref{sec:3} by showing that a transformation
implemented in $\ck_{n,k}$ is in fact a partially defined construction.
We then show that the unitary matrix used to represent the transformation
is not optimal in terms of number of elementary gates needed to decompose 
it. 
We propose a different construction that indeed requires lesser
number of elementary gates and we also argue its optimality w.r.t the Universal gate set. 

In Section~\ref{sec:4} we use the construction to improve the gate count of
the circuit $\ck_{n,k}$ in a generalized manner.
We remove the redundant gates in the circuit and analyze the different
partially defined transformations implemented in the circuit to further reduce the 
gate counts of the circuit. We denote the improved circuit for preparing any Dicke 
State $\ket{D^n_k}$ as $\ckt_{n,k}$. To the best of our knowledge this is the most 
optimal implementation of a deterministic Dicke state preparation circuit for $\ket{D^n_k},~2 \leq k \leq n-2$.

Next in Section~\ref{sec:5} we discuss the architectural constraints posed by
the current generation Quantum Computers that are available for public use 
through different cloud services. We discuss the restrictions in terms of 
implementing CNOT gates between two qubits in an architecture and how it 
increases the number of CNOT gates needed to implement a circuit in an architecture.  
In this regard we show that the improvements described by us in Section~\ref{sec:4}
not only reduces gate counts but also reduces architectural constraints. 

We implement the circuits $\ck_{4,2}$ and $\ckt_{4,2}$ on the IBM-QX machine 
``ibmqx2''\cite{ibmq} and calculate the deviation in each case from ideal measurement statistics
using a simple error measure. 
Next we show how two circuits with the same number of CNOT gates and 
the same architectural restrictions can lead to different expected error 
due to different CNOT distribution across the qubits. We analyze this 
by proposing modifications in the circuit $\ckt_{4,2}$ possible because 
partial nature of certain transformations and how it reduces the number of 
CNOT gates functioning erroneously on expectation in a fairly generalized error model. 
We finish this section by drawing out the general CNOT map of $\ckt_{n,k}$,
shown as the graph $G^{n,k}$ and observing that there in fact exists $n-k-1$ independent 
modifications each leading to a different CNOT distribution. 

We conclude the paper in Section~\ref{sec:6} by describing  
the future direction of work in this domain and also note down
open problems in this area that we feel will improve 
our understanding both in the domains of partially defined transformations 
and architectural constraints. 

\section{Preliminaries}
\label{sec:2}
We first define some terminologies that we frequently use before
moving onto some definitions and the preliminaries.

\subsection{Notations}
\begin{enumerate}
\item $\ket{v_2}$: If we look at a system with $n$ qubits 
then all the $2^n$ orthogonal states in the computational basis
can be expressed as $\ket{b_1b_2\ldots b_n},~ b_i \in \{0,1\}, 1 \leq i \leq n$.

In that case for representing the state $\ket{b_1b_2\ldots b_n}$ 
we treat it as a binary string and express it as $\ket{v_2}$
where $v=\displaystyle \sum_{i=1}^n b_i2^{n-i}$.

\item  $R_y(\theta)$: The $R_y$ gate is a single qubit gate defined as follows. 
$R_y(\theta) \equiv e^{-\theta Y}=
\begin{bmatrix}
\ncost{2} & -\nsint{2} \\
\nsint{2} & \ncost{2}
\end{bmatrix}
$.

\item $X$: This is a single qubit gate defined as 
$X=
\begin{bmatrix}
0 & 1 \\
1 & 0
\end{bmatrix}
$.

\item $CU^i_j$: While implementing a controlled unitary on a two qubit 
subsystem we use the following notations.
Let there be a $n$-qubit system.
$CU^i_j$ represents a two qubit controlled unitary operation 
where the $i$-th qubit is the control qubit and the $j$-th qubit is the target qubit.
\end{enumerate}

\subsection{Maximally Partial Unitary Transformation}
Let there be a unitary transformation that acts on $n$ qubits. 
To perform this transformation we have to create a corresponding 
unitary matrix.
If the transformation is defined for all $2^n$ states of some 
orthonormal basis then the unitary matrix is completely defined.
On the other hand if the transformation is defined for 
a single state belonging to the computational basis, only a single 
column of the corresponding $2^n \times 2^n$ matrix is filled.
The rest can be filled up conveniently, provided its unitary 
property is satisfied. In  this regard we call a unitary transformation
on $n$ qubits to be maximally partial if it is defined for 
$2^n-1$ states of some orthonormal basis. That implies only
a column of the matrix is not defined. 
In this paper we observe how corresponding to a maximally partial unitary
transformation there can be multiple unitary matrices
and how the minimal number of elementary gates needed to implement these matrices
may vary.

We end this section by describing the structure of Dicke states
and a circuit designed for its preparation. 
\subsection{The Dicke State Preparation Circuit $\ck_{n,k}$}
The circuit $\ck_{n,k}$ as described in~\cite{dicke} works on the
$n$ qubit system $\ket{q_1q_2\ldots q_n}$.
The circuit $\ck_{n,k}$ is broken into $n-1$ blocks of the form $SCS^x_y$
of which the first $n-k$ blocks are of the form 
$SCS^{n-t}_k,~n-t>k$ which is then followed by $k-1$ 
blocks of the form $SCS^i_{i-1},k \geq i \geq 2$. 

A block $SCS^n_k$ consists of a two qubit transformation and $k-1$ three qubit 
transformations. The two qubit transformation works on the $n-1$ and $n$-th qubits
and we denote it as $\mu_n$.
We describe the overall structure of the circuit again in Section~\ref{sec:5}. 

The three qubit transformations are of the form $\M_n^l, n-1 \leq i \leq n-k+1$
where $\M_l^n$ works on the qubits $l-1,l$ and $n$.
This construction is interesting in how the transformations $\mu$ and $\M$ 
are partially defined which raises different implementation choices,
with possibly different number of gates needed for elemental decomposition.
We now describe these two transformations for reference.
We denote by $\ket{ab}_x$ the qubits in the $x-1$ and $x$-th position in a system.

\begin{align*}
\mu_n: \quad
& \ket{00}_n \rightarrow \ket{00}_n \nonumber \\
& \ket{11}_n \rightarrow \ket{11}_n \nonumber \\
& \ket{01}_n \rightarrow \sqrt{\frac{1}{n}}\ket{01}_n+\sqrt{\frac{n-1}{n}}\ket{10}_n \nonumber
\end{align*}
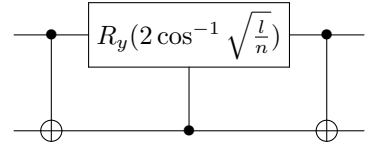
\begin{figure}[!ht]
\centering
\Qcircuit @C=1em @R=2em {
& & & & & \quad &\ctrl{1} &\gate{R_y(2 \cos^{-1}{\sqrt{\frac{l}{n}})}} &\ctrl{1} & \qw \\
& & & & & \quad &\targ &\ctrl{-1} &\targ & \qw \\
}
\caption{Implementation of $\mu_n$}
\label{fig:c3}
\end{figure}
\begin{align*}
\M^l_n: \quad
& \ket{00}_l\ket{0}_n \rightarrow \ket{00}_l\ket{0}_n \nonumber \\
& \ket{01}_l\ket{0}_n \rightarrow \ket{01}_l\ket{0}_n \nonumber \\
& \ket{00}_l\ket{1}_n \rightarrow \ket{00}_l\ket{1}_n \nonumber \\
& \ket{11}_l\ket{1}_n \rightarrow \ket{11}_l\ket{1}_n \nonumber \\
& \ket{01}_l\ket{1}_n \rightarrow \sqrt{\frac{n-l+1}{n}}\ket{01}_l\ket{1}_n \\
& ~~~~~~~~~~~~~~~~~~~~~~~+\sqrt{\frac{l-1}{n}}\ket{11}_l\ket{0}_n \nonumber
\end{align*}
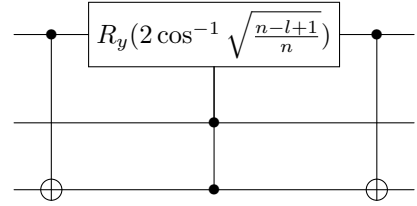
\begin{figure}[!ht]
\centering
\Qcircuit @C=1em @R=2em {
& & & & & & \ctrl{2} & \gate{R_y(2\cos^{-1}\sqrt{\frac{n-l+1}{n}})} &\ctrl{2} & \qw \\
& & & & & & \qw      & \ctrl{-1}                                & \qw     & \qw  \\
& & & & & & \targ    & \ctrl{-2}                                &\targ    & \qw \\
}
\caption{Implementation of $\M^l_n$}
\label{fig:c4}
\end{figure}

The implementations of these transformations in~\cite{dicke} is shown in
Figure~\ref{fig:c3} and~\ref{fig:c4} respectively.
The first transformation, $\mu_n$ is in fact a maximally partial unitary transform.
Because of the partially defined nature of the transformation
the $CR_y$ and $CCR_y$ gates are also not fed all possible 
inputs. 
Instead the input to the $CR_y$ gates is only from the 
subspace spanned by the computational basis states $\ket{00},\ket{10}$ and $\ket{01}$.
Similarly the input to the  $CCR_y$ gate is only from the subspace
spanned by the states $\ket{000}, \ket{010}, \ket{001}, \ket{011},\text{ and } \ket{110}$.

Next in Section~\ref{sec:3} we look how partially defined transformations
can be implemented more efficiently, and argue the optimality of this
improvement with respect to this particular building block. 
Then in Section~\ref{sec:4} we reduce the gate count of the circuit 
$\ck_{n,k}$ by removing redundancies and analyzing how the $\mu$ and $\M$
transformations act only on a subset of the defined computational basis 
states in specific cases. 

\section{Example of Optimality for a Maximally Partial Unitary Transformation}
\label{sec:3}
We have described the two partially defined unitary transformations used in the
circuit $\ck_{n,k}$.
The implementation of the first transformation, $\mu_n$ is done using a controlled $R_y$
gate and two CNOT gates.
This $CR_y$ gate only acts on the states
$\ket{00}, \ket{10}, \ket{01}$ and their superpositions and 
the transformation never acts on the $\ket{11}$ state.
If we take $\theta=2\cos^{-1}\big(\sqrt{\frac{1}{n}}\big)$.
We denote the transformation implemented by the $CR_y(\theta)$ gate
on the defined basis states as $T_1(\theta)$,
and the corresponding transformation is as follows:
\begin{align}
\label{eq:tr1}
T_1(\theta): \quad
&\ket{00} \rightarrow \ket{00} \\ \nonumber
&\ket{10} \rightarrow \ket{10} \\ \nonumber
&\ket{01} \rightarrow   \big( \cos(\frac{\theta}{2})\ket{0}+ \sin(\frac{\theta}{2})\ket{1} \big ) \ket{1} \nonumber
\end{align}
This is in fact a maximally defined partial unitary transformation.
While the gate $CR_Y(\theta)$ can perform this transformation,
it needs at least $4$ elementary gates to implement. We first
prove this necessary requirement using an important result from 
\cite[Theorem B]{song}, which we note down for reference.

\begin{theorem}\cite{song}
\label{th:opt}
\begin{enumerate}
\item For a controlled gate $CU$ if
$tr (UX)=0,~ tr(U) \neq 0,~ \det U=1,~ U \neq \pm I$ then 
the minimal number of elementary gates needed to implement 
$CU$ is $4$.

\item For a controlled gate $CU$ if
$tr(U)=0,~ \det U=-1, ~ U \neq \pm X$ then 
the minimal number of elementary gates needed to implement 
$CU$ is $3$.

\item For a controlled gate $CU$ the minimal 
number of number of elementary gates needed to implement 
$CU$ is less than three $3$ iff $U \in \{e^{i\phi}I,e^{i\phi}X,e^{i\phi}Z \},~0 \leq \phi \leq 2\pi$.
\end{enumerate}
\end{theorem}

Our lemma follows immediately.
\begin{lemma}\label{th:cry4}
It takes minimum $4$ elementary gates to implement the $CR_y(\theta)$ gate.
\end{lemma}
\begin{proof}
We calculate the values of $\det R_y(\theta)$ and $tr(R_y(\theta)X)$
to confirm the minimal number of gates needed to decompose $CR_y(\theta)$.
\begin{align*}
& \det R_y(\theta)= \sin^2(\frac{\theta}{2})+\cos^2(\frac{\theta}{2})=1 \\
& R_y(\theta)X=
\begin{bmatrix}
-\nsint{2} & \ncost{2} \\
 \ncost{2} & \nsint{2} 
\end{bmatrix} \implies tr(R_y(\theta)X)=0 
\end{align*}
The result ($1$) of Theorem~\ref{th:opt} concludes the proof.
\end{proof}
However the transformation $T_1(\theta)$ can in fact be implemented using
three elementary gates as follows.
$$T_1(\theta) \equiv \Big( R_y( \frac{-\alpha}{2}) \otimes I_2 \Big) {\sf CNOT}^2_1
 \Big( R_y( \frac{\alpha}{2}) \otimes I_2 \Big),~ \frac{\alpha}{2}=\frac{\pi}{2}-\frac{\theta}{2} $$
This decomposition has also been used by Cruz et.al~\cite{dn1} in the $W_n$ ($D^n_1$) 
state preparation algorithm. However, the corresponding transformation is defined only for the states 
$\ket{00}$ and $\ket{01}$ and no insight into the optimality of the implementation is given.

We first derive the underlying $4 \times 4$ unitary matrix $U^0(\alpha)$ that describes this three gate
transformation. Next we prove that $U^0(\alpha)$ needs at least three gates to be implemented 
by verifying the conditions of result (2) of Theorem~\ref{th:opt}.
We end this section by showing that the transformation $T_1(\theta)$ needs at least 
three elementary gates (including one CNOT) to be implemented, proving the optimality 
of the $U^0(\alpha)$ implementation.

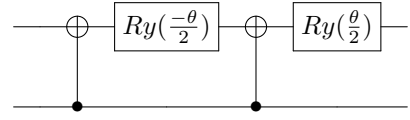
\begin{figure}[!ht]
\centering
\Qcircuit @C=1em @R=2em {
&&&& & \qw &\targ &\gate{Ry(\frac{-\theta}{2})} &\targ &\gate{Ry(\frac{\theta}{2})} & \qw \\
&&&& & \qw &\ctrl{-1} & \qw &\ctrl{-1} & \qw & \qw  
}
\caption{Implementation of $CR_y(\theta)$}
\label{fig:c0}
\end{figure}

\begin{figure}[!ht]
\centering
\Qcircuit @C=1em @R=2em {
&&&& & \qw &\gate{Ry(\frac{\alpha}{2})} &\targ &\gate{Ry(\frac{-\alpha}{2})} & \qw & \qw  \\
&&&& & \qw & \qw &\ctrl{-1} & \qw & \qw & \qw\\
}
\caption{Implementation of $U^o(\alpha)$}
\label{fig:c1}
\end{figure}
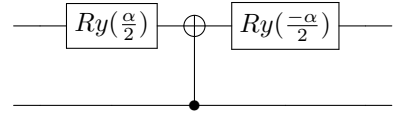

\begin{theorem}
The gate $U^0(\alpha)$ performs the partially defined unitary transformation $T_1(\theta)$
where $\alpha=\pi-\theta$
and needs minimum three elementary gates to be implemented.
\end{theorem}
\begin{proof}
We first study the transformation carried out by $U^0$ in the subspace of $T_1$.
\begin{align*}
\ket{00} & \\
&\xrightarrow{R_y \left( \frac{\alpha}{2}\right)\text{ on }q1}
\Big ( \ncosa{4}\ket{0}+\nsina{4}\ket{1} \Big ) \ket{0}  \\
&\xrightarrow{{\sf CNOT}_1^2}
\Big ( \ncosa{4}\ket{0}+\nsina{4}\ket{1} \Big ) \ket{0} \\
&\xrightarrow{R_y(-\frac{\alpha}{2})\text{ on }q1}
\Big( \ncosa{4}\big(\ncosa{4}\ket{0}-\nsina{4})\ket{1}\big)+\\
&~~~~~~~~~~~~~~~~~\nsina{4})\big(\nsina{4})\ket{0}+\ncosa{4})\ket{1} \big) \Big) \ket{0} \\
&= \big( \cos^2(\frac{\alpha}{4})+\sin^2(\frac{\alpha}{4}) \big)\ket{00}=\ket{00} \\
\ket{10} &\\
& \xrightarrow{R_y \left( \frac{\alpha}{2}\right)\text{ on }q1}
\Big ( -\nsina{4}\ket{0}+\ncosa{4}\ket{1} \Big ) \ket{0}  \\
& \xrightarrow{{\sf CNOT}_1^2} 
\Big ( -\nsina{4}\ket{0}+\ncosa{4}\ket{1} \Big ) \ket{0} \\
& \xrightarrow{R_y(-\frac{\alpha}{2})\text{ on }q1} \Big(
-\nsina{4}\big(\ncosa{4}\ket{0}-\nsina{4})\ket{1}\big)+ \\
&~~~~~~~~~~~~~~~~~\ncosa{4})\big(\nsina{4})\ket{0}+\ncosa{4})\ket{1} \big) \Big) \ket{0} \\
&= \big( \cos^2(\frac{\alpha}{4})+\sin^2(\frac{\alpha}{4}) \big)\ket{10}=\ket{10} \\
 \ket{01} & \\
& \xrightarrow{R_y \left( \frac{\alpha}{2}\right)\text{ on }q1}
\Big ( \ncosa{4}\ket{0}+\nsina{4}\ket{1} \Big ) \ket{1}  \\
& \xrightarrow{{\sf CNOT}_1^2} 
\Big ( \nsina{4}\ket{0}+\ncosa{4}\ket{1} \Big ) \ket{1} \\
& \xrightarrow{R_y(-\frac{\alpha}{2})\text{ on }q1} \Big(
\nsina{4}\big(\ncosa{4}\ket{0}-\nsina{4})\ket{1}\big)+ \\
&~~~~~~~~~~~~~~~~~\ncosa{4})\big(\nsina{4})\ket{0}+\ncosa{4})\ket{1} \big) \Big) \ket{1} \\
&= \big( 2\ncosa{4}\nsina{4} \ket{0}+ (\cos^2(\frac{\alpha}{4})-\sin^2(\frac{\alpha}{4})\ket{1} \big)\ket{1} \\
&=\Big( \nsina{2}\ket{0}+\ncosa{2}\ket{1} \Big)\ket{1}
\end{align*}

Setting $\alpha=\pi-\theta$ gives us the same transformation as defined by $T_1(\theta)$.

Now we completely define the gate $U^0$ by studying the transformation
acted on the state $\ket{11}$.
\begin{align*}
 \ket{11} &\\
&\xrightarrow{R_y \left( \frac{\alpha}{2}\right)\text{ on }q1} 
\Big ( -\nsina{4}\ket{0}+\ncosa{4}\ket{1} \Big ) \ket{0} \\  
& \xrightarrow{{\sf CNOT}_1^2} \Big ( \ncosa{4}\ket{0}-\nsina{4}\ket{1} \Big ) \ket{0} \\
& \xrightarrow{R_y(-\frac{\alpha}{2})\text{ on }q1} \Big(
\ncosa{4}\big(\ncosa{4}\ket{0}-\nsina{4})\ket{1}\big)- \\
&~~~~~~~~~~~~~~~~~ \nsina{4})\big(\nsina{4})\ket{0}+\ncosa{4})\ket{1} \big) \Big) \ket{0} \\
&= \big( (\cos^2(\frac{\alpha}{4})-\sin^2(\frac{\alpha}{4})\ket{0}- 2\ncosa{4}\nsina{4} \ket{0}  \big)\ket{1} \\
&=\Big( \ncosa{2}\ket{0}-\nsina{2}\ket{1} \Big)\ket{1}
\end{align*}
So the overall transformation provided by $U^0(\alpha)$ is:
\begin{align*}
& \ket{00} \rightarrow \ket{00} \\
& \ket{10} \rightarrow \ket{10} \\
& \ket{01} \rightarrow \Big( \nsina{2}\ket{0}+\ncosa{2}\ket{1} \Big)\ket{1} \\
& \ket{11} \rightarrow \Big( \ncosa{2}\ket{0}-\nsina{2}\ket{1} \Big)\ket{1}
\end{align*}
Therefore the gate $U^0(\alpha)$ is a two qubit gate which can be expressed as 
a controlled gate $CU(\alpha)$ gate where $U(\alpha)=
\begin{bmatrix}
\nsina{2} & \ncosa{2}  \\
\ncosa{2} & -\nsina{2} \\
\end{bmatrix}
$.
Now $tr U(\alpha)=0$ and $\det U(\alpha)=-1$ for all $\alpha$.
Therefore we can conclude from the result ($2$) in Theorem~\ref{th:opt} that this gate requires at least
three gates to be implemented.
\end{proof}
We finally show the optimality of this implementation for implementing the 
two qubit transformation $T_1(\theta)$.
\begin{lemma}
The transformation $T_1(\theta)$ needs at least one CNOT
and two single qubit gates to be implemented for $0 <\theta <\pi$.
\end{lemma}
\begin{proof}
The transformation $T_1(\theta)$ is only defined for the basis states $\ket{00},\ket{01}$ and
$\ket{10}$. 
Any matrix $M(\theta)$ that can carry out the transformation is of the form
$\begin{bmatrix}
1 & 0 & 0 & a \\ 
0 & \cos(\frac{\theta}{2}) & 0 & b \\
0 & 0 & 1 & c \\
0 & \sin(\frac{\theta}{2}) & 0 & d
\end{bmatrix}$
where $a,b,c,d$ are complex unknowns.
However since $M(\theta)$ is unitary we have $M(\theta)M^{\dagger}(\theta)=I$.
Therefore $1+aa^*=1 \implies a=0$ and $1+cc^*=1 \implies c=0$.
That is the matrix $M_{\theta}$ is a controlled unitary $CM_1(\theta)$
and $M_1(\theta)=
\begin{bmatrix}
\cos(\frac{\theta}{2}) & b \\
\sin(\frac{\theta}{2}) & d
\end{bmatrix}$.

Now for $0 < \theta < \pi$ both $\cos(\frac{\theta}{2})$ and $\sin(\frac{\theta}{2})$
are non zero. This implies that the matrix cannot fulfill the necessary conditions 
defined in result (3) of Theorem~\ref{th:opt} and therefore cannot be expressed with
less than three elementary gates. 
\end{proof}

Now we use our observations to improve the gate count of the circuit $\ck_{n,k}$.
\section{Improved gate counts for Circuits of $\ket{D^n_k}$}
\label{sec:4} 

We first count the number of CNOT and single qubit gates
in $\ck_{n,k}$ by reviewing the circuit. 
The circuit is composed of $n-1$ blocks of gates called $SCS$.
There are $n-k-1$ blocks of the form $SCS^t_k,~ k<t \leq n$
and $k-1$ blocks of the form $SCS^{i+1}_i,~1 \leq i \leq k-1$.

Each block $SCS^t_k$ consists of one two qubit transformation $\mu_t$ 
which is implemented on the qubits $t-1$ and $t$
and $k-1$ three qubit transformations of the type $\M^l_t,~t-1 \leq l \leq t-k-2$.
Here $\mu_t$ is implemented on the $t-1$ and $t$-th qubit and $\M^l_t$ is implemented on the
$l-1$, $l$ and $t$-th qubit, as described in Section~\ref{sec:2}. 
Each transformation of type $\mu$ is decomposed into two CNOT and a $T_1(\theta)$
transformation which is implemented as a $CR_y$ gate by adjusting the value of $\theta$.
We have shown in Lemma~\ref{th:cry4} that a $CR_y$ transformation needs minimum
$4$ gates to implement. In fact it needs at least two CNOT gates.
Therefore each $\mu$ transformation needs four CNOT and two single
qubit gate.
The number of transformations of type $\M^l_n$ is 
\begin{align*}
&(n-k)(k-1)+ \displaystyle \sum_{i=1}^{k-2}i \\
&=nk-n+k-k^2 + \frac{(k-1)(k-2)}{2} \\
&=nk- \frac{k(k+1)}{2}-n+1.
\end{align*}
Each $\M^l_n$ transformation is shown to require six CNOT and four single qubit gates.
However one CNOT gate of for each $\M^l_n$ transformation can be canceled by rearranging the 
first two CNOT gates of the next transformation.

The total number of CNOT gates and single qubit gates used to prepare
the state $\ket{D^n_k}$ is shown in Table~\ref{tab:cnot1}.

\begin{table}[!ht]
\begin{center}
\begin{tabular}{| c | c |}
\hline
CNOT gates & $5(nk- \frac{k(k+1)}{2}-n+1)+4(n-1)$ \\ \hline
single qubit gates & $4(nk- \frac{k(k+1)}{2}-n+1)+2(n-1)$ \\ \hline
\end{tabular}
\caption{Gates needed to prepare $\ket{D^n_k}$ as in~\cite{dicke} }
\label{tab:cnot1}
\end{center}
\end{table}

Figure~\ref{fig:cir63_0} shows the circuit $\ck_{6,3}$ in terms of CNOT, $CR_y$ and $CCR_y$ gates. 

\begin{figure*}[!ht]
\begin{center}
\scriptsize
\vfill
\hspace*{0.5cm}
\Qcircuit @C=0.27em @R=0.1em {
\lstick{\ket{0}}& \qw& \qw \qw &\qw &\qw&\qw &\qw&\qw &\qw&\qw &\qw&\qw &\qw&\qw &\qw&\qw &\qw&\qw&\qw &\qw&\qw &\qw&\qw &\qw&\qw&\qw&\qw&\ctrl{3}&\gate{\sqrt{\frac{3}{4}}}&\ctrl{3}&\qw&\qw&\qw&\ctrl{2}&\gate{\sqrt{\frac{2}{3}}}&\ctrl{2}&\ctrl{1}&\gate{\sqrt{\frac{1}{2}}}&\ctrl{1}&\meter\\
\lstick{\ket{0}}& \qw& \qw &\qw &\qw&\qw &\qw&\qw &\qw&\qw &\qw&\qw &\qw&\qw &\qw&\qw &\qw&\qw &\ctrl{3}&\gate{\sqrt{\frac{3}{5}}} &\ctrl{3}&\qw&\qw&\qw&\ctrl{2}&\gate{\sqrt{\frac{2}{4}}}&\ctrl{2}&\qw&\ctrl{-1}&\qw &\ctrl{1}&\gate{\sqrt{\frac{1}{3}}} &\ctrl{1}&\qw&\ctrl{-1}&\qw&\targ&\ctrl{-1}&\targ&\meter\\
\lstick{\ket{0}}& \qw& \qw &\qw &\qw &\qw &\qw &\qw &\qw &\ctrl{3} &\gate{\sqrt{\frac{3}{6}}} &\ctrl{3} &\qw &\qw&\qw &\ctrl{2} &\gate{\sqrt{\frac{2}{5}}} &\ctrl{2} &\qw &\ctrl{-1} &\qw&\ctrl{1}&\gate{\sqrt{\frac{1}{4}}}&\ctrl{1}&\qw&\ctrl{-1}&\qw&\qw&\qw&\qw&\targ&\ctrl{-1}&\targ&\targ&\ctrl{-1}&\targ&\qw&\qw&\qw&\meter\\
\lstick{\ket{0}}& \qw &\gate{X}& \qw &\qw &\qw &\ctrl{2} &\gate{\sqrt{\frac{2}{6}}} &\ctrl{2} &\qw &\ctrl{-1} &\qw &\ctrl{1} &\gate{\sqrt{\frac{1}{5}}} &\ctrl{1} &\qw&\ctrl{-1}&\qw&\qw&\qw&\qw&\targ&\ctrl{-1}&\targ &\targ&\ctrl{-1}&\targ&\targ&\ctrl{-2}&\targ&\qw&\qw&\qw&\qw&\qw&\qw&\qw&\qw&\qw&\meter\\
\lstick{\ket{0}}& \qw &\gate{X} &\ctrl{1} &\gate{\sqrt{\frac{1}{6}}}&\ctrl{1} &\qw &\ctrl{-1} &\qw &\qw &\qw &\qw&\targ &\ctrl{-1} &\targ &\targ &\ctrl{-1} &\targ &\targ &\ctrl{-2}&\targ&\qw&\qw&\qw&\qw&\qw&\qw&\qw&\qw&\qw&\qw&\qw&\qw&\qw&\qw&\qw&\qw&\qw&\qw&\meter\\
\lstick{\ket{0}}& \qw &\gate{X} &\targ &\ctrl{-1} &\targ&\targ &\ctrl{-1} &\targ &\targ &\ctrl{-2} &\targ&\qw&\qw&\qw&\qw&\qw&\qw&\qw&\qw &\qw&\qw&\qw&\qw&\qw&\qw&\qw&\qw&\qw&\qw&\qw&\qw&\qw&\qw&\qw&\qw&\qw&\qw&\qw&\meter\\
}
\end{center}
\caption{Description of the circuit $\ck_{6,3}$}
\label{fig:cir63_0}
\end{figure*}

We show the circuit of $\ket{D^4_2}$ formed according to this construction method
in Figure~\ref{fig:cir63_0}.

We now improve the gate counts of the circuit $\ck_{n,k}$ in the following ways.

\subsection*{Replacing $CR_y$ with $CU$}
We first use the $CU$ gate shown in Section~\ref{sec:3} to implement the transformation $T_1$ corresponding
to each $\mu$ transformation which needs one CNOT and two single qubit gates to be implemented. 
Therefore each $\mu$ transformation needs three CNOT and two single qubit gates.
Since there are $n-1$ $\mu$ transformations this reduces 
the number of CNOT by $n-1$ for any $\ket{D^n_k}$.

Next we observe that some of the $\mu_n$ and $\M^l_n$ transformations
act as identity transformation, which we count as a function of $k$
for any $\ket{D^n_k}$.
\subsection*{The $\mu$ and $\M$ transformations that act like Identity}
Let there be a $n$ qubit system in some state $\ket{\phi}$.
This state can be uniquely represented as a superposition of all $2^n$ 
(computational) basis state. 
The amplitude of a particular basis state may or may not be zero depending
on the description of $\ket{\phi}$. We call a basis state with 
non zero amplitude an active basis state.
The affect of a unitary transformation $T$ on this state can be completely described by
observing how it transforms the active basis states of $\ket{\phi}$.
If the $k$-th qubit is in the zero (one) state in all the active basis states
and the transformation $T$ doesn't act on the $k$-th qubit in non trivial way on any
of those basis states then the $k$-th qubit in all the
active basis states of $T\ket{\phi}$ will also be in the zero(one)
state.
Using this simple fact we prove the following theorem using induction.

\begin{theorem}
If the $n$ qubit system is expressed as superposition of  
computational basis states after the block $SCS^{n-t}_k$ has acted
then it can be expressed as 
\begin{align*}
\sum\limits_{a=0}^{2^{t+1}-1} ~ \sum\limits_{b=0}^{2^{t+1}-1}\alpha_{a,b}\Big( &\ket{0}^{\otimes n-k-1-t} 
\big(\bigotimes_{i=1}^{t+1}\ket{a^{bin}_i} \big)\\
&\ket{1}^{\otimes k-1-t}
\big( \bigotimes_{j=1}^{t+1}\ket{b^{bin}_j} \big)\Big).
\end{align*}
\end{theorem}

\begin{proof}
The statement implies that the first $n-k-t-1$ qubits 
are all in the state $\ket{0}$ and the $(n-k+1)$-th qubit
and the next $k-t-1$ qubits are all in the state $\ket{1}$
in all active basis states of the $n$ qubit system after the
block $SCS^{n-t}_k$ has acted on it.

We first prove the statement for $t=0$.
The $n$ qubit system is at first in the state 
$\ket{\psi_0}=\ket{0}^{\otimes (n-k)}\ket{1}^{\otimes k}$
and the block to be applied is $SCS^n_k$.
This block consists of the gates $\mu_n$, $\M^{i}_n,n-1 \leq i \leq n-k+1$.
We know that the transformation $\mu_n$ affects the $(n-1)$ and the $n$-th qubit 
and the transformation $\M^l_n$ affects the $(l-1)$ and $n$ th qubit. Additionally
$\mu_n$ acts as identity on a basis state if the $n-q$th qubit of the
basis state is in the $\ket{1}$ state. Similarly 
$\M^l_n$ acts as identity on a basis state
if the $(n-l-1)$-th qubit is in the $\ket{1}$ state.

The last $k$ qubits of $\ket{\psi_0}$ are in the state $\ket{1}$
and therefore $\mu_n$ and $\M^{n-1}_n, \M^{n-2}_n, \ldots \M^{n-k+2}_n$
act as identity transformations.
Finally the transformation $\M^{n-k+2}_n$ is applied. 
The first qubit to this transformation is in the state $\ket{0}$
and therefore this transformation may lead to basis states with 
either $\ket{0}$ or $\ket{1}$ in the $(n-k)$-th and $n$-th positions.
Therefore the resultant state can be written as 
$$ \displaystyle \sum_{a_1,a_2 \in \{0,1\}} \alpha_{a_1a_2} 
\ket{0}^{\otimes n-k-1}\ket{a_1}\ket{1}^{\otimes k-1}\ket{a_2}.$$

Thus the first $n-k-1$ qubits are all in the state $\ket{0}$ and
the $(n-k+1)$-th qubit and the next $k-1$ qubits are all in the $\ket{1}$ 
state in all active basis states of the system.This concludes the base case. 

Now assuming that our statement holds true for some $t-1<k-2$ 
we show that the statement also holds for $t$.
The $SCS^{n-t}_k$ block is composed of the transformations
$\mu_{n-t}$, $\M^{i}_{n-t},n-t-1 \leq t \leq n-t-k+1$.
The $n$ qubit system is in the state
\begin{align*}
\ket{\psi_{t-1}} 
=\sum\limits_{a=0}^{2^{t}-1} ~ \sum\limits_{b=0}^{2^{t}-1} \alpha_{a,b} &\Big(\ket{0}^{\otimes n-k-t} 
\big(\bigotimes_{i=1}^{t}\ket{a^{bin}_i} \big) \\
& \ket{1}^{\otimes k-t}
\big( \bigotimes_{j=1}^{t}\ket{b^{bin}_j} \big)\Big).
\end{align*}
That is, the first $n-k-t$ qubits are in the state $\ket{0}$ in all active basis states
and the $n-k+1$ and the next $k-t-1$ qubits are in the state $\ket{1}$.
These are the first qubits to the transformations
$\mu_{n-t}$,$\M^{i}_{n-t}, n-t-1 \leq i \leq n-k$.
This implies the $\mu$ transformation and the $(k-2-t)$ $\M$ transformations
act as identity transformations on all active basis states.

The next $t$ $\M$ transformations may get the $\ket{0}$ state as the first
qubit and therefore the $n$-qubit system
before the last $\M$ has been applied is in the state
\begin{align*}
\ket{\psi_{t-1}'}
=\sum\limits_{a=0}^{2^{t}-1} ~ \sum\limits_{b=0}^{2^{t+1}-1} \alpha_{a,b}\Big( &\ket{0}^{\otimes n-k-t} 
\big(\bigotimes_{i=1}^{t}\ket{a^{bin}_i} \big) \\
&\ket{1}^{\otimes k-1-t}
\big( \bigotimes_{j=1}^{t+1}\ket{b^{bin}_j} \big)\Big).
\end{align*}
Finally the last three qubit transformation of  the block $SCS^{n-t}_k$
$\M^{n-t-k+1}_{n-t}$ acts on the system. 
Now since the $(n-t-k)$-th qubit is in the state $\ket{0}$ in all active basis states,
the $\M$ gate may non trivially act on it and the $(n-t)$-th qubit.
This results in the state
\begin{align*}
\ket{\psi_t}
=\sum\limits_{a=0}^{2^{t+1}-1} ~ \sum\limits_{b=0}^{2^{t+1}-1} \alpha_{a,b}\Big( &\ket{0}^{\otimes n-k-t-1} 
\big(\bigotimes_{i=1}^{t+1}\ket{a^{bin}_i} \big) \\
&\ket{1}^{\otimes k-1-t}
\big( \bigotimes_{j=1}^{t+1}\ket{b^{bin}_j} \big)\Big).
\end{align*}
This completes the proof.
It is important to note that there may be many basis states in the expression of $\ket{\psi_t}$
with zero amplitude. However our focus is on qubits that are definitely going to be either in the
zero state or in the one state in all active basis states.
\end{proof}

This proof also shows that the $\mu$ transformation and the $k-2-t$ $\M$ transformations
in the block $SCS^{n-t}_k,t<k-1$ act as identity transformations and 
therefore can be removed from the circuit, which is the second improvement. 
Therefore the number of $\mu$ transformations omitted is $k-1$ and 
the number of $\M$ transformation omitted are $\frac{(k-2)(k-1)}{2}$.
This removes $3(k-1)+\frac{5(k-2)(k-1)}{2}$ CNOT and 
$2(k-1)+\frac{4(k-2)(k-1)}{2}$ single qubit gates.

\subsection*{The first non identity $\M$ transformation in $SCS^n_k$}
Having removed the $\mu$ and $\M$ transformations we now
look at the first transformation of the block $SCS^{n-t}_k,t<k-1$
which is $\M^{n-k+1}_{n-t}$.
This transformation depends on the state of the $n-k,n-k+1$ and $(n-t)$-th qubits
and affects the state of the $(n-k)$-th and the $(n-t)$-th qubit.
At this stage the $n=$ qubit system is at the state
\begin{small}
$$
\sum\limits_{a=0}^{2^{t}-1} ~ \sum\limits_{b=0}^{2^{t}-1} \alpha_{a,b} \ket{0}^{\otimes n-k-t} 
\Big(\bigotimes_{i=1}^{t}\ket{a^{bin}_i} \Big) \ket{1}^{\otimes k-t}
\Big( \bigotimes_{j=1}^{t}\ket{b^{bin}_j} \Big).$$
\end{small}

Therefore in all the active basis states both the $n-k$th and the $n-t$th 
qubits are in the state $\ket{1}$.
Therefore the three qubit transformation applied by $\M^{n-k+1}_{n-t}$
can be expressed as follows substituting $l=n-k+1$:
\begin{align*}
& \ket{11}_{l}\ket{1}_{n-t} \rightarrow \ket{11}_{l}\ket{1}_{n-t} \\
& \ket{01}_{l}\ket{1}_{n-t} \rightarrow 
\sqrt{\frac{n-t-l+1}{n-t}}\ket{01}_{l}\ket{1}_{n-t} \\
&~~~~~~~~~~~~~~~~~~~~~~~~~~~~~~+\sqrt{\frac{l-1}{n-t}}\ket{11}_{l}\ket{0}_{n-t}.
\end{align*}
This is in-fact can be implemented as a 
a two qubit transformation of the type $\mu$.
as the $(n-k+1)$-th qubit is in the $\ket{1}$ 
state in all the active basis states.

The transformation acts on the $(n-k)$-th and $(n-t)$-th qubits as 
$\M^{n-k+1}_{n-t} \equiv   ({\sf CNOT}^{n-k}_{n-t})(CU^{n-t}_{n-k}(\theta)) ({\sf CNOT}^{n-k}_{n-t})$
where $\theta =2 \cos^{-1}\Big(\sqrt{\frac{n-t-l+1}{n-t}} \Big)$. 
We know that the $CU$ gate requires one CNOT and two $R_y$ gates to be implemented
therefore $\M^{n-k+1}_{n-t}$ requires only three CNOT and two $R_y$ gates.
This improvement is reflected for all $SCS^{n-t}_k$ such that $n-t \geq n-k+2$ that is for
$0 \leq t \leq k-2$.
Therefore it reduces the number of CNOT gate in the circuit by further $2(k-1)$
and the number of single qubit gates by $2(k-1)$ as well.

Additionally for $\ket{D^n_k},~k>1$ when $SCS^n_k$ is applied the $n$-qubit 
system is in the state $\ket{0}^{\otimes n-k}\ket{1}^{\otimes k}$
and therefore the transformation $\M^{n-k+1}_{n-t}$
only acts on the basis state
$\ket{011}$. The corresponding transformation is 
\begin{align*}
\ket{01}_{n-k+1}\ket{1}_{n} \rightarrow &
\sqrt{\frac{k}{n}}\ket{01}_{n-k+1}\ket{1}_{n} \\
&~~~~+\sqrt{\frac{n-k}{n}}\ket{11}_{n-k+1}\ket{0}_{n}.
\end{align*}
This can be implemented using a $R_y(\cos^{-1}\sqrt{\frac{k}{n}})$ on the $(n-k)$-th qubit
followed by a CNOT gate ${\sf CNOT}^{n-k}_n$
which removes further two CNOT and one single qubit gate.
\begin{center}
\begin{figure*}[!h]
\hspace*{1cm}
\Qcircuit @C=0.2em @R=0.2em {
\lstick{\ket{0}}  &\qw &\qw&\qw &\qw&\qw &\qw&\qw &\qw&\qw&\qw &\qw&\qw&\qw &\qw&\qw &\qw&\qw &\qw&\qw&\qw&\qw&\ctrl{3}&\gate{\sqrt{\frac{3}{4}}}&\ctrl{3}&\qw&\qw&\qw&\ctrl{2}&\gate{\sqrt{\frac{2}{3}}}&\ctrl{2}&\ctrl{1}&\gate{\sqrt{\frac{1}{2}}}&\ctrl{1}&\meter\\
\lstick{\ket{0}} &\qw&\qw &\qw&\qw &\qw&\qw &\qw&\qw &\qw&\qw &\qw&\qw &\ctrl{3}&\gate{\sqrt{\frac{3}{5}}} &\ctrl{3}&\qw&\qw&\qw&\ctrl{2}&\gate{\sqrt{\frac{2}{4}}}&\ctrl{2}&\qw&\ctrl{-1}&\qw&\ctrl{1}&\gate{\sqrt{\frac{1}{3}}} &\ctrl{1}&\qw&\ctrl{-1}&\qw&\targ&\ctrl{-1}&\targ&\meter\\
\lstick{\ket{0}}  &\qw &\qw &\qw &\gate{\sqrt{\frac{3}{6}}}&\ctrl{3}&\qw &\qw &\qw&\qw &\ctrl{2} &\gate{\sqrt{\frac{2}{5}}} &\ctrl{2} &\qw &\ctrl{-1} &\qw&\ctrl{1}&\gate{\sqrt{\frac{1}{4}}}&\ctrl{1}&\qw&\ctrl{-1}&\qw&\qw&\qw&\qw&\targ&\ctrl{-1}&\targ&\targ&\ctrl{-1}&\targ&\qw&\qw&\qw&\meter\\
\lstick{\ket{0}}& \qw &\gate{X} &\qw &\qw &\qw&\qw &\qw &\qw &\qw &\qw&\qw&\qw&\qw&\qw&\qw&\targ&\ctrl{-1}&\targ &\targ&\ctrl{-1}&\targ&\targ&\ctrl{-2}&\targ&\qw&\qw&\qw&\qw&\qw&\qw&\qw&\qw&\qw&\meter\\
\lstick{\ket{0}}& \qw &\gate{X} &\qw &\qw &\qw &\qw&\qw &\qw &\qw &\targ &\ctrl{-2} &\targ &\targ &\ctrl{-2}&\targ&\qw&\qw&\qw&\qw&\qw&\qw&\qw&\qw&\qw&\qw&\qw&\qw&\qw&\qw&\qw&\qw&\qw&\qw&\meter\\
\lstick{\ket{0}}& \qw &\gate{X}  &\qw &\qw &\targ&\qw&\qw&\qw&\qw&\qw&\qw&\qw&\qw &\qw&\qw&\qw&\qw&\qw&\qw&\qw&\qw&\qw&\qw&\qw&\qw&\qw&\qw&\qw&\qw&\qw&\qw&\qw&\qw&\meter\\
}
\caption{Description of the Circuit $\ckt_{6,3}$}
\label{fig:cir63_1}
\end{figure*}
\end{center}
We denote this circuit by $\ckt_{n,k}$. Figure~\ref{fig:cir63_1} shows the structure of $\ckt_{6,3}$.
Combining the results we get the following count of CNOT and single
qubit gates in the improved circuit.
We now calculate the total improvement in the CNOT and single qubit gate counts 
for the $\ket{D^n_k},k>1$ preparation circuit $\ckt_{n,k}$. 
\begin{itemize}
\item The total number of CNOT gates removed $=n-1+3(k-1)+\frac{5(k-2)(k-1)}{2}+2(k-1)+2$.

\item The total number of single qubit gates removed $=2(k-1)+\frac{4(k-2)(k-1)}{2}+2(k-1)+1$. 
\end{itemize}
Therefore the total number of CNOT gates present in the circuit is
\begin{align*} 
& 5(nk- \frac{k(k+1)}{2})-n+1  \\
&-\Big(n-1+3(k-1)+\frac{5(k-2)(k-1)}{2}+2(k-1) +2 \Big) \\
&=5nk-5k^2-2n
\end{align*}
The number of single qubit gates present in the circuit is 
\begin{align*}
& 4(nk- \frac{k(k+1)}{2}-n+1)+2(n-1) \\
&-\Big(2(k-1)+\frac{4(k-2)(k-1)}{2}+2(k-1)+1 \Big) \\
& = 4nk-4k^2-2n+1
\end{align*}

For $k=1$ We get the number of CNOT as $3n-3$ (from $n-1$ $\mu$ transformations) 
and the number of single qubits gate as $2n-2$.
However, one CNOT gate can be further removed from each $\mu$ gate as the
active basis states in input to the $\mu$ transformations are only $\ket{00}$ and $\ket{01}$.
The resultant circuit is identical to the linear $W_n$ preparation circuit in~\cite{dn1} 
and contains $2n-2$ CNOT and $2n-2$ single qubit gates and thus we don't elaborate it further. 

We know that the state $\ket{D^n_k}$ can be prepared by first forming the state $\ket{D^n_{n-k}}$ and then
applying a X gate to each qubit. On that note it is interesting to observe that after the improvements 
the Circuits for $\ket{D^n_k}$ and $\ket{D^n_{n-k}}$ require the same number of CNOT gates. 

\begin{table}[H]
\centering
\begin{tabular}{|c|c|c|c|c|c|c|}
\hline 
\diagbox{$k$}{$n$} & 4 & 5 & 6 & 7& 8 \\ \hline
2& 22,12& 31,20 &40,28 &49,36 &58,44  \\ \hline
3& 27,7 & 41,20 &55,33 &69,46 &83,59  \\ \hline
4&      & 46,10 &65,28 &84,46 &103,64 \\ \hline
5&      &       &70,13 &94,36 &118,59 \\ \hline
6&      &       &      &99,16 &128,44 \\ \hline
7&      &       &      &      &133,19 \\ \hline
\end{tabular}
\caption{CNOT gate count of the pair $\ck_{n,k},\ckt_{n,k}$}
\label{tab:cnotcomp}
\end{table}

\begin{table}[H]
\centering
\begin{tabular}{|c|c|c|c|c|c|c|}
\hline 
\diagbox{$k$}{$n$} & 4 & 5 & 6 & 7& 8 \\ \hline
2&14,9  &20,15 & 26,21 &32,27 & 38,33 \\ \hline
3&18,5  &28,15 & 38,25 &48,35 & 58,45 \\ \hline
4&      &32,7  & 46,21 &60,35 & 74,49 \\ \hline
5&      &      & 50,9  &68,27 & 86,45 \\ \hline
6&      &      &       &72,11 & 94,33	\\ \hline
7&      &      &       &      & 98,13	\\ \hline
\end{tabular}
\caption{Single qubit gate count of the pair $\ck_{n,k},\ckt_{n,k}$}
\label{tab:scomp}
\end{table}

Table~\ref{tab:cnotcomp} and~\ref{tab:scomp} show the number of CNOT and
single qubit gates needed to implement
the states $\ket{D^n_k}$ for $4 \leq n \leq 8, 1 \leq k \leq n-1$, respectively.

In the next section we show that 
our observations not only reduces the gate counts of the circuit 
but also reduces its architectural constraints.
\section{Actual Implementation and architectural constraints}
\label{sec:5}
\subsection{Architectural Constraints}
We are at the stage where quantum circuits can be implemented on actual quantum computers 
using cloud services, such as IBM Quantum Experience, also known as IQX~\cite{ibmq}.
However the architecture of the individual back-end quantum machines pose restrictions to implementation
of a particular circuit. The most prominent constraint is that of the CNOT implementation.
In this regard we use the terms architectural constraint and CNOT constraint interchangeably. 
Every quantum system $Q$ with $n$ (physical) qubits has a CNOT map, which we express as $G^Q_A(V^Q,E^Q)$
where $V^Q=\{q_1,q_2,\ldots q_n\}$.
In this graph the nodes represent the qubits and the edges represent CNOT implementability.
A directed edge $q_i \rightarrow q_j$ implies that a CNOT can be implemented
with $q_i$ as control and $q_j$ as target in the system $Q$.
The edges in the CNOT maps of all the publicly available IQX machines are bidirectional.

Let there be a circuit $\ck$ on $n$ qubits. We denote the (logical) qubits 
$c_1,c_2,\ldots c_n$. We also have a CNOT map corresponding to the circuit, 
which describes the CNOT gates used in the circuit. We describe this as the directed graph 
$G_{\ck}(V^{\ck},E^{\ck})$ where $V^{\ck}=\{ c_1,c_2, \ldots c_m \}$
and there is a directed edge $c_i \rightarrow c_j$ if there is one or more CNOT with
$c_i$ as control and $c_j$ as target.

Therefore if the graph $G_{\ck}$ can shown to be the subgraph of $G^Q_A$ by some mapping of the 
logical qubits to the physical qubits then the circuit $\ck$ can be implemented 
on the architecture with the same number of CNOT gates. However, if such a map is not possible,
then the circuit can be implemented on that architecture either by using
swap gates which require additional CNOT gates or changing the construction of the
circuit.
There are mapping solutions such as the one applied by IQX which dynamically changes 
the structure of the circuit to implement a circuit in an architecture that does not
meet the circuit's CNOT constraints.
Similarly in their paper Zulehner et.al~\cite{efmap} have also proposed an efficient mapping solution. 
However it is not always possible to avoid an increase in the number of CNOT gates.
Given a circuit it is crucial to find it's minimal architectural needs
in terms of the CNOT map without increasing the CNOT gate count.
The IBM-Q systems mapping solutions show the modified circuit as the transpiled 
circuit given any circuit as input, although its solutions are not always optimal. 
In this paper we first consider the system ``${\sf ibmqx2}$'' ($Q_1$) of IQX.
The CNOT map of $Q_1$ is shown in the Figure~\ref{fig:cmapqx2}.
\begin{figure}[H]
\begin{center}
\begin{tikzpicture}[-latex ,auto ,node distance =1.5 cm and 1.5cm ,on grid ,
semithick ,
state/.style ={ circle ,top color =white , bottom color = white ,
draw,black, text=black , minimum width =0.5 cm}]
\node[state] (E) {$1$};
\node[state] (A)[above=of E] {${2}$};
\node[state] (B) [right=of E] {${3}$};
\node[state] (C) [left=of E] {$0$};
\node[state] (D) [below=of E] {$4$};

\path (A) edge (E);\path (E) edge (A);
\path (B) edge (E);\path (E) edge (B);
\path (C) edge (E);\path (E) edge (C);
\path (D) edge (E);\path (E) edge (D);
\path (A) edge (C);\path (C) edge (A);
\path (D) edge (B);\path (B) edge (D);

\end{tikzpicture}
\end{center}
\caption{The CNOT map of $Q_1$ represented as $G^{Q_1}_A$}
\label{fig:cmapqx2}
\end{figure}

Against this backdrop we observe the CNOT constraints of the circuit $\ck_{4,2}$,
implemented to prepare $\ket{D^4_2}$.
Then we implement the circuit $\ckt_{4,2}$ which is the result of the improvements
shown in Section~\ref{sec:4}. We observe that the improvement proposed by us not only
reduces gate counts but also reduces CNOT constraints. 
We implement these circuits in the system $Q_1$ and
compare the measurement statistics of the two circuits by measuring the deviation from
the ideal measurement statistics and find that the results of $\ckt_{4,2}$ is much more
closely aligned with the ideal results.
We end this section by showing how some changes in the circuit $\ckt_{4,2}$ possible 
because of partially defined transformations can lower the error in the circuit due to 
CNOT on expectation without a reduction in number of CNOT gates or
change in the CNOT constraints.

\subsection{Implementation and Improvement for $\ket{D^4_2}$}
We start by constructing the circuit $\ck_{4,2}$.
We implement every $CR_y$ gate using two CNOT gates and two $R_y$ gate
as we know that the $CR_y$ gate needs at least $4$ gates to be implemented 
and every three qubit $\M^l_n$ transformation using five CNOT and four $R_y$ gates (as given in 
the description of~\cite{dicke}).  
The resultant circuit $\ck_{4,2}$ is shown in Figure~\ref{fig:cir42_0}.
This circuit contains $22$ CNOT gates.
We use the notation $\theta^x_y$ to denote the angle $2\cos^{-1}(\sqrt{\frac{x}{y}})$.

The CNOT map of the circuit is shown in Figure~\ref{fig:cmapstage1}.
\begin{figure*}[!h]
\centering
\footnotesize
\hspace*{0.5cm}
\Qcircuit @C=0.2em @R=1em {
\lstick{\ket{0}}& \qw & \qw & \qw & \qw & \qw & \qw & \qw & \qw & \qw & \qw & \qw & \qw & \qw & \qw & \qw & \qw & \qw & \qw & \qw & \qw & \qw & \qw & \qw &\targ &\ctrl{2} &\gate{\frac{-\theta^{2}_{3}}{4}}~ &\targ ~ &\gate{\frac{\theta^{2}_{3}}{4}} ~ &\targ &\gate{\frac{-\theta^{2}_{3}}{4}} &\targ &\gate{\frac{\theta^{2}_{3}}{4}} &\ctrl{2} &\ctrl{1} &\targ &\gate{\frac{-\theta^{1}_{2}}{2}}&\targ &\gate{\frac{\theta^{1}_{2}}{2}} &\ctrl{1} &\qw &\meter\\
\lstick{\ket{0}}& \qw & \qw & \qw & \qw & \qw & \qw & \qw & \qw & \targ & \ctrl{2} & \gate{\frac{-\theta^{2}_{4}}{4}} &\targ &\gate{\frac{\theta^{2}_{4}}{4}} &\targ &\gate{\frac{-\theta^{2}_{4}}{4}} &\targ &\gate{\frac{\theta^{2}_{4}}{4}} &\ctrl{2}  &\ctrl{1} &\targ &\gate{\frac{-\theta^{1}_{3}}{2}} &\targ &\gate{\frac{\theta^{1}_{3}}{2}} &\ctrl{-1} & \qw & \qw & \qw & \qw &\ctrl{-1} & \qw & \qw & \qw & \qw &\targ &\ctrl{-1} & \qw  &\ctrl{-1}& \qw &\targ& \qw &\meter \\
\lstick{\ket{0}}& \qw & \gate{X} & \qw & \ctrl{1} & \targ & \gate{\frac{-\theta^{1}_{4}}{2}} & \targ & \gate{\frac{\theta^{1}_{4}}{2}} & \ctrl{-1} & \qw & \qw & \qw & \qw & \ctrl{-1} & \qw & \qw & \qw &\qw &\targ &\ctrl{-1} &\qw &\ctrl{-1} &\qw&\qw &\targ &\qw &\ctrl{-2} &\qw &\qw  &\qw &\ctrl{-2} &\qw &\targ & \qw & \qw & \qw & \qw & \qw & \qw &\qw &\meter \\
\lstick{\ket{0}}& \qw & \gate{X} & \qw & \targ & \ctrl{-1}& \qw & \ctrl{-1} & \qw & \qw &\targ & \qw &\ctrl{-2}& \qw & \qw & \qw &\ctrl{-2} & \qw &\targ & \qw & \qw & \qw & \qw & \qw & \qw & \qw & \qw & \qw & \qw & \qw & \qw & \qw & \qw & \qw & \qw & \qw & \qw & \qw & \qw &\qw&\qw &\meter \\
}
\caption{The circuit $\ck_{4,2}$ due to~\cite{dicke}}
\label{fig:cir42_0}
\end{figure*}

\begin{figure}[H]
\begin{center}
\begin{tikzpicture}[-latex ,auto ,node distance =1.5 cm and 1.5cm ,on grid ,
semithick ,
state/.style ={ circle ,top color =white , bottom color = white ,
draw,black, text=black , minimum width =0.5 cm}]
\node[state] (A) {${q_0}$};
\node[state] (B) [right=of A] {${q_1}$};
\node[state] (C) [below=of A] {$q_2$};
\node[state] (D) [below=of B] {$q_3$};

\path (A) edge (B);\path (B) edge (A);
\path (B) edge (C);\path (C) edge (B);
\path (C) edge (D);\path (D) edge (C);
\path (D) edge (B);\path (B) edge (D);
\path (A) edge (C);\path (C) edge (A);

\end{tikzpicture}
\end{center}
\caption{The CNOT map of $\ck_{4,2}$}
\label{fig:cmapstage1}
\end{figure}
We then implement $\ckt_{4,2}$ by making the following changes to $\ck_{4,2}$.
\begin{enumerate}
\item Implement the $CU^o$ gate instead of $CR_y$ gates.

\item Remove the Redundant $\mu$ and $\M$ transformation.

\item Reduce the gate count in implementation of $\M^{n-k+1}_{n-t}$ type transformations.
\end{enumerate}
This brings the total number of CNOT gates in the circuit to $12$. 
We name the circuit at this stage $\ckt_{4,2}$.

These steps not only reduce the CNOT gates in the circuit but also reduces the CNOT constraints
of the circuit. 
Figure~\ref{fig:cir42_3} shows the circuit at this stage and the reduced CNOT map $G_{\ckt_{4,2}}$ 
as shown in the Figure~\ref{fig:cmapstage3}.

\begin{figure*}[!h]
\centering
\hspace*{0.3cm}
\Qcircuit @C=0.2em @R=1em {
&&&&&& & & & &\lstick{\ket{0}} & \qw & \qw & \qw & \qw & \qw & \qw & \qw & \qw &\targ &\ctrl{2} &\gate{\frac{-\theta^{2}_{3}}{4}} &\targ &\gate{\frac{\theta^{2}_{3}}{4}} &\targ &\gate{\frac{-\theta^{2}_{3}}{4}} &\targ &\gate{\frac{\theta^{2}_{3}}{4}} &\ctrl{2} &\ctrl{1} &\qw &\gate{\frac{\pi}{2}-\frac{\theta^{1}_{2}}{2}}&\targ &\gate{-(\frac{\pi}{2}-\frac{\theta^{1}_{2}}{2})} &\ctrl{1} &\qw &\meter\\
&&&&&& & & & &\lstick{\ket{0}} & \gate{\theta^{2}_{4}} & \qw &\ctrl{2}  &\ctrl{1} &\qw &\gate{\frac{\pi}{2}-\frac{\theta^{1}_{3}}{2}} &\targ &\gate{-(\frac{\pi}{2}-\frac{\theta^{1}_{3}}{2})} &\ctrl{-1} & \qw & \qw & \qw & \qw &\ctrl{-1} & \qw & \qw & \qw & \qw &\targ &\qw & \qw  &\ctrl{-1}& \qw &\targ& \qw &\meter \\
&&&&&& & & & &\lstick{\ket{0}} & \gate{X}  \qw & \qw &\qw &\targ &\qw &\qw &\ctrl{-1} &\qw&\qw &\targ &\qw &\ctrl{-2} &\qw &\qw  &\qw &\ctrl{-2} &\qw &\targ & \qw & \qw & \qw & \qw & \qw & \qw &\qw &\meter \\
&&&&&& & & & &\lstick{\ket{0}} & \gate{X} & \qw &\targ & \qw & \qw & \qw & \qw & \qw & \qw & \qw & \qw & \qw & \qw & \qw & \qw & \qw & \qw & \qw & \qw & \qw & \qw & \qw & \qw &\qw&\qw &\meter \\
}
\caption{The Circuit $\ckt_{4,2}$}
\label{fig:cir42_3}
\end{figure*}

\begin{figure}[H]
\begin{center}
\begin{tikzpicture}[-latex ,auto ,node distance =1.5 cm and 1.5cm ,on grid ,
semithick ,
state/.style ={ circle ,top color =white , bottom color = white ,
draw,black, text=black , minimum width =0.5 cm}]
\node[state] (A) {${q_0}$};
\node[state] (B) [right=of A] {${q_1}$};
\node[state] (C) [below=of A] {$q_2$};
\node[state] (D) [below=of B] {$q_3$};

\path (A) edge (B);\path (B) edge (A);
\path (B) edge (C);\path (C) edge (B);
\path (B) edge (D);
\path (A) edge (C);\path (C) edge (A);

\end{tikzpicture}
\caption{The CNOT map $G^{\ckt_{4,2}}$ corresponding to the circuit $\ckt_{4,2}$}
\label{fig:cmapstage3}
\end{center}
\end{figure}
In fact the Graph $G^{\ckt_{4,2}}$ can be shown to be a subgraph
of $G^{Q_1}_A$ under several mappings of qubits. Therefore this
circuit can be implemented in the {\sf ``ibmqx2''} ($Q_1$) machine with $12$
CNOT.
However the CNOT constraints of the circuits corresponding to
even $D^5_k,k>1$ cannot be met by any IBM-Q architecture at this stage. 
Now we compare the results of the circuits $\ck_{4,2}$ which is due to~\cite{dicke} 
and $\ckt_{4,2}$ which is what we obtained after the reductions and modifications.

\subsection*{Comparison of Measurement Statistics of $\ck_{4,2}$ and $Cir_{4,2}^1$}
The output by an ideal Quantum Computer would produce the
state $\sqrt{ \frac{1}{{n \choose w}}}\displaystyle \sum_{wt(i)=w} \ket{i_2}$
on a correct $\ket{D^n_w}$ preparation circuit.
We first verify the resultant state vectors to of the two circuits to see
that they both ideally produce 
$\sqrt{\frac{1}{6}} \Big(\ket{0011}+\ket{0101}+\ket{0110}+\ket{1100}+\ket{1010}+\ket{1001} \Big)$
and then use a primary error measure based on measurement in computational basis 
to estimate the closeness of the states formed by the two circuits from the ideal state $\ket{D^4_2}$.

We run both the circuits for the maximum possible shots $(8192)$ and use the measurement 
statistics to estimate the closeness to the desired state using the following error measure.
We define our error measure $EM_{n,w}$ for the Dicke state $\ket{D^n_w}$ as follows.
Let $p_i$ be the percentage of times the measurement of the circuit $\ck$ yields the result $i_2$.
Then we have
$$ EM_{n,w}(\ck)= \frac{1}{2} \Big( \displaystyle \sum_{j , wt(j)=w} \abs*{p_j-\frac{1}{{n \choose w}}} +  \displaystyle \sum_{j , wt(j) \neq w} p_j \Big)$$
\begin{figure*}[!ht]
\begin{center}
\includegraphics[scale=0.29]{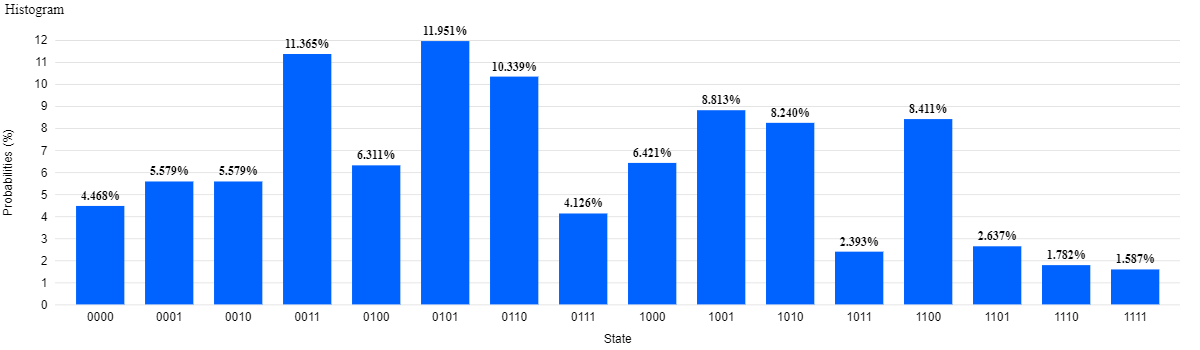}
\caption{Measurement statistics for $\ck_{4,2}$ with $25$ CNOT in transpiled circuit}
\label{fig:D42basic}
\end{center}
\end{figure*}

An $EM$ value of $0$ signifies that the measurement statistics are exactly aligned 
with the expected ideal results while the $EM$ value can at maximum be $1$.
We have calculated the $EM$ values for results of different mappings for the circuit $\ck_{4,2}$.

It is very interesting to see that under different mapping of logical qubits to physical qubits in $Q_1$ from 
the user end the IQX mapping solution provided different transpiled circuits. 
We know that the number of CNOT in $\ck_{4,2}$ is $22$. The transpiled circuits
for $\ck_{4,2}$ had a minimum of $25$ CNOT gates and were as high as $31$ in some cases.
The corresponding transpiled circuit contains $25$ CNOT gates which is the least of all the 
transpiled circuits.
Figure~\ref{fig:D42basic} shows the measurement statistics corresponding to the circuit with the 
minimum $EM$ value, which is equal to $0.4088$.

\begin{figure*}[!ht]
\begin{center}
\includegraphics[scale=0.29]{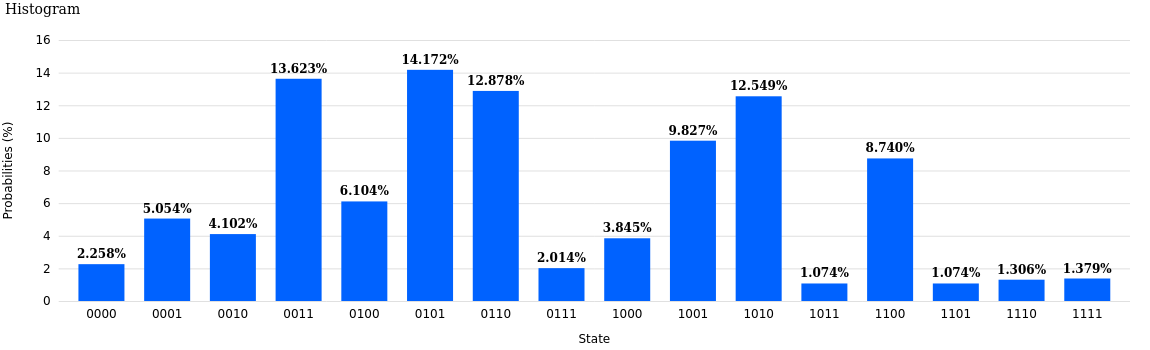}
\caption{Measurement statistics for $\ckt_{4,2}$ corresponding to the map $M_1$}
\label{fig:D42opt}
\end{center}
\end{figure*}

Next we look at the measurement statistics of the circuit $\ckt_{4,2}$.
There are many mappings between logical and physical qubits in this case 
such that the CNOT constraint of the circuit is met. 
Let such a map be $M:\{q_0,q_2,q_3,q_4\} \rightarrow \{0,1,2,3,4\}$. 
Then if there is a CNOT between $q_i$ and $q_j$ then there is an 
edge $M(q_i) \leftrightarrow M(q_j)$ in graph $G^{Q_1}_A$.
In such mappings the IQX mapping solution didn't implement any 
modification in the transpiled circuit as expected.

Here we present the result for the following map $M_1$
\begin{align*}
M1:\quad &  q_0 \rightarrow 3,\quad   q_1 \rightarrow 2,\quad q_2 \rightarrow 4,\quad   q_3 \rightarrow 0.
\end{align*}
Figure \ref{fig:D42opt} shows the measurement statistics corresponding 
to this mapping and the resultant $EM$ value is $0.282103$.
These results show that the circuit $\ck_{4,2}$ needs more than the 
specified number of CNOT while being implemented on ``ibmqx2'' and
the measurement statistics of $\ckt_{4,2}$ is much more closely
aligned with the ideal measurement statistics compared to $\ck_{4,2}$.

\subsection{Modifications leading to different CNOT error distributions}
\label{subsec:1}
We now discuss how we can in fact use partially defined transformations
to further fine tune the circuit $\ckt_{4,2}$ depending on the 
specifications of the architecture. 
We consider a four qubit architecture $A_4$ with the same CNOT
connectivity as $G^{\ckt_{4,2}}$ and only differs in CNOT error distribution.
We then observe how further modifying the circuit $\ckt_{4,2}$
can lead to lower CNOT error on expectation against some 
error distributions in the architecture $A_4$.
We assume every edge  in the CNOT map of $A_4$ is bidirectional,
as is the case with all currently publicly available IBM-Q machines.

The CNOT error rate when applying a CNOT between qubits $i$ and $j$
(such that the edge $i \leftrightarrow j$ is present in $G_A$) is denoted 
as $e_{ij}$.Figure~\ref{fig:arch4} shows the CNOT map of the architecture.

\subsubsection*{CNOT error model}
In this regard we define our error model to calculate CNOT error on expectation
of a circuit implemented in the architecture $A_4$.
The probability of a CNOT placed between qubits $i$ and $j$
acting erroneously in a circuit is dependent on the error rate 
of the corresponding CNOT coupling in the architecture.
We call this CNOT error. 
We denote this probability with 
$f_e:[0,1] \rightarrow [0,1]$. 
We do not assume the exact nature of $f_e$, 
but only that it is directly proportional to error rate 
(i.e. an increasing function) which is by definition. 

Next we define the following Bernoulli random variables to calculate the  
the number of CNOT acting erroneously on expectation
when a circuit is applied on this architecture.  
We define a variable $x_k$ corresponding to each CNOT used in a circuit.
The variable is assigned zero if the $k$-th CNOT is applied correctly while
executing a circuit, and one otherwise.

Let us suppose the $k$-th CNOT is applied between qubits $i$ and $j$.
Then we have $Pr(x_k=1)=f_e(e_{ij})$ and
The expected error while applying the CNOT is $\mathbb{E}(x_k)=f_e(e_{ij})$.
Therefore the CNOT error on expectation 
while implementing a circuit $\ck$ on the architecture is 
$$\mathbb{E}(\ck)=\displaystyle \sum v_{ij} f_e(e_{ij}).$$

Having described the error model we look at the CNOT distribution of the 
circuit $\ckt_{4,2}$ as a weighted graph $G_f$. The vertices and the edges of this graph
is same as that of $G^{\ckt_{4,2}}$.
The weight of an edge $q_i \leftrightarrow q_j$ is the number of CNOT gates 
applied between the two qubits in the circuit. The graph $G_f$ is shown in 
Figure~\ref{fig:cmapdis}.

Now we implement the circuit $\ckt_{4,2}$ on the architecture $A_4$
so that all CNOT constraints can be met.
We observe that only the qubit $1$ has degree $3$ and therefore $q_1$
is mapped to the physical qubit $1$. 
Then we can have the following maps which satisfies all the CNOT constraints.
\begin{enumerate}
\item $q_1\rightarrow 1$,~$q_0\rightarrow 0$,~$q_2\rightarrow 2$,~$q_3\rightarrow 3$.
\item $q_1\rightarrow 1$,~$q_0\rightarrow 2$,~$q_2\rightarrow 0$,~$q_3\rightarrow 3$.
\end{enumerate}
Then the expected CNOT error of the circuit $\ckt_{4,2}$ when applied on the architecture $A_4$ is
$$\mathbb{E}(\ckt_{4,2})= 5f_e(e_{01})+3f_e(e_{02})+3f_e(e_{12})+ f_e(e_{13}).$$

We now show the circuit $\ckt_{4,2}$ (described in Figure~\ref{fig:cir42_3}) can be further modified using partially defined transformations
so that the  CNOT error in the circuit on this architecture will reduce on expectation under some error distribution
conditions. 

\begin{figure}[H]
\centering
\begin{tikzpicture}[-latex ,auto ,node distance =1.5 cm and 1.5cm ,on grid ,
semithick ,
state/.style ={ circle ,top color =white , bottom color = white ,
draw,black, text=black , minimum width =0.5 cm}]
\node[state] (A) {${q_0}$};
\node[state] (B) [right=of A] {${q_1}$};
\node[state] (C) [below=of A] {$q_2$};
\node[state] (D) [below=of B] {$q_3$};

\path (A) edge (B);\path (B) edge node[above] {$5$} (A);
\path (B) edge (C);\path (C) edge node[left] {$3$} (B);
\path (D) edge (B);\path (B) edge node[right] {$1$} (D);
\path (A) edge (C);\path (C) edge node[left] {$3$} (A);

\end{tikzpicture}
\caption{Graph $G_f$ corresponding to $G^{\ckt_{4,2}}$.}
\label{fig:cmapdis}
\end{figure}

\begin{figure}[H]
\centering
\begin{tikzpicture}[-latex ,auto ,node distance =1.5 cm and 1.5cm ,on grid ,
semithick ,
state/.style ={ circle ,top color =white , bottom color = white ,
draw,black, text=black , minimum width =0.5 cm}]
\node[state] (A) {$0$};
\node[state] (B) [right=of A] {$1$};
\node[state] (C) [below=of A] {$2$};
\node[state] (D) [below=of B] {$3$};

\path (A) edge (B);\path (B) edge node[above] {$e_{01}$} (A);
\path (B) edge (C);\path (C) edge node[left] {$e_{12}$} (B);
\path (D) edge (B);\path (B) edge node[right] {$e_{13}$} (D);
\path (A) edge (C);\path (C) edge node[left] {$e_{02}$} (A);

\end{tikzpicture}
\caption{CNOT map of the architecture $A$.}
\label{fig:arch4}
\end{figure}
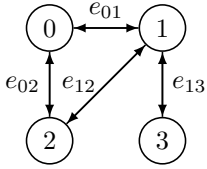

\begin{figure}[H]
\centering
\begin{tikzpicture}[-latex ,auto ,node distance =1.5 cm and 1.5cm ,on grid ,
semithick ,
state/.style ={ circle ,top color =white , bottom color = white ,
draw,black, text=black , minimum width =0.5 cm}]
\node[state] (A) {${q_0}$};
\node[state] (B) [right=of A] {${q_1}$};
\node[state] (C) [below=of A] {$q_2$};
\node[state] (D) [below=of B] {$q_3$};

\path (A) edge (B);\path (B) edge node[above] {$5$} (A);
\path (B) edge (C);\path (C) edge node[left] {$3$} (B);
\path (C) edge (D);\path (D) edge node[above] {$1$} (C);
\path (A) edge (C);\path (C) edge node[left] {$3$}(A);

\end{tikzpicture}
\caption{Graph $G_f'$ corresponding to $\ckt_{4,2}'$}
\label{fig:cmapmod2}
\end{figure}

\begin{figure*}[!ht]
\centering
\hspace*{0.2cm}
\Qcircuit @C=0.2em @R=1em {
& & & && & & &\quad \quad & \lstick{\ket{0}} & \qw & \qw & \qw &\qw & \qw & \qw & \qw & \qw & \qw &\targ &\ctrl{2} &\gate{\frac{-\theta^{2}_{3}}{4}} &\targ &\gate{\frac{\theta^{2}_{3}}{4}} &\targ &\gate{\frac{-\theta^{2}_{3}}{4}} &\targ &\gate{\frac{\theta^{2}_{3}}{4}} &\ctrl{2} &\ctrl{1} &\qw &\gate{\frac{\pi}{2}-\frac{\theta^{1}_{2}}{2}}&\targ &\gate{-(\frac{\pi}{2}-\frac{\theta^{1}_{2}}{2})} &\ctrl{1} &\qw &\meter\\
& & & && & & &\quad \quad & \lstick{\ket{0}} & \gate{\theta^{2}_{4}} & \qw &\ctrl{1}  &\qw & \qw & \qw &\gate{\frac{\pi}{2}-\frac{\theta^{1}_{3}}{2}} &\targ &\gate{-(\frac{\pi}{2}-\frac{\theta^{1}_{3}}{2})} &\ctrl{-1} & \qw & \qw & \qw & \qw &\ctrl{-1} & \qw & \qw & \qw & \qw &\targ &\qw & \qw  &\ctrl{-1}& \qw &\targ& \qw &\meter \\
& & & && & & &\quad \quad & \lstick{\ket{0}} & \gate{X}  & \qw & \targ & \qw &\ctrl{1} &\qw & \qw &\ctrl{-1} &\qw&\qw &\targ &\qw &\ctrl{-2} &\qw &\qw  &\qw &\ctrl{-2} &\qw &\targ & \qw & \qw & \qw & \qw & \qw & \qw &\qw &\meter \\
& & & && & & &\quad \quad & \lstick{\ket{0}} & \qw & \qw & \qw & \qw & \targ & \qw & \qw & \qw & \qw & \qw & \qw & \qw & \qw & \qw & \qw & \qw & \qw & \qw & \qw & \qw & \qw & \qw & \qw & \qw &\qw&\qw &\meter \\
}
\caption{Circuit description of $\ckt_{4,2}'$ }
\label{fig:cir42_4}
\end{figure*}

The first $R_y$ gate that acts on the second qubit of $\ckt_{4,2}$
is followed by the CNOT gates ${\sf CNOT}^2_3$ and ${\sf CNOT}^2_4$.
The combined transformation $T_4$ of these two CNOT is defined 
only for two basis states on $4$ qubits 
$\ket{0011} \rightarrow \ket{0011}$ and $\ket{0111} \rightarrow \ket{0100}$.

We use the partial nature of the transformation to modify the circuit as follows.
Note that transformation $T_4$ is the first transformation that acts on $q_3$.
Then if we start the circuit from the state $\ket{0010}$ instead of $\ket{0011}$
then we can define the transformation $T_4^1$ such that
\begin{align*}
& T_4^1 \equiv (I_2 \otimes {\sf CNOT}^3_2 \otimes I_2) (I_2 \otimes I_2 \otimes {\sf CNOT}^2_1) \\
\implies & T_4^1\ket{0010}=\ket{0011},~~ T_4^1\ket{0110}=\ket{0100}
\end{align*}
resulting in the same output states as $T_4$ for all the computational
basis states for which $T_4$ is defined. It is important to note that this
implementation would not have been possible if the transformation was defined 
for all the $8$ basis states of the second third and fourth qubits.

We denote this circuit as $\ckt_{4,2}'$ and it is drawn in Figure~\ref{fig:cir42_4}. 
We denote the weighted CNOT map of the circuit $\ckt_{4,2}'$ as $G_f'$ and
it is shown in Figure~\ref{fig:cmapmod2}.

We now see that in $G_f'$ $q_2$ has degree three and therefore any map that meets
all the CNOT constraints will have $q_2 \rightarrow 1$.
Therefore we can have the following maps that satisfies all the CNOT constraints.

\begin{enumerate}
\item $q_2 \rightarrow 1,~~ q_0 \rightarrow 0,~~q_1 \rightarrow 2,~~q_3 \rightarrow 3 $.
\item $q_2 \rightarrow 1,~~ q_0 \rightarrow 2,~~q_1 \rightarrow 0,~~q_3 \rightarrow 3 $.  
\end{enumerate}
The CNOT error on expectation for both the circuits is 
$$\mathbb{E}(\ckt_{4,2}')= 5f_e(e_{02})+3f_e(e_{01})+3f_e(e_{12})+ f_e(e_{13}).$$
Now we calculate the conditions when $\mathbb{E}(\ckt_{4,2}')$ is less than $\mathbb{E}(\ckt_{4,2}')$.
\begin{align*}
& \mathbb{E}(\ckt_{4,2}')<\mathbb{E}(\ckt_{4,2}) \\
\implies & 5f_e(e_{02})+3f_e(e_{01})+3f_e(e_{12})+ f_e(e_{13}) \\
&< 5f_e(e_{01})+3f_e(e_{02})+3f_e(e_{12})+ f_e(e_{13}) \\
\implies & f_e(e_{02})< f_e(e_{01}) \\
\implies & e_{02} < e_{01}.
\end{align*}
This gives us an insight into how different CNOT distributions in a circuit may lead to better results
without reduction in the number of CNOT gates or a reduction in the architectural constraints. 
We conclude this section by describing the architectural constraint of the circuit 
$\ckt_{n,k}$.

\subsection{The CNOT map of $\ckt_{n,k}$}
The CNOT gates in the circuit $\ckt_{n,k}$ are due to 
implementation of the $\mu$ and $\M$ transformation of the different $SCS^n_k$ blocks.  
$\mu_n$ forms an edge in the CNOT map of the form $n-1 \leftrightarrow n$. 
where as $\M^l_n$ forms the edges $(l-1) \leftrightarrow n$ and  $l \rightarrow (l-1)$.
However in the circuit $\ckt_{n,k}$ the transformations $\M^{n-k+1}_t$ do not have a
CNOT between the neighboring qubits $n-k$ and $n-k+1$. 

We divide the edges into two groups. One corresponding to CNOT gates between neighboring 
qubits and one where the positions of the qubits differ at least by two.
We calculate the edges of each of these types.
\begin{itemize}
\item The neighboring qubits with CNOT connections are the qubits $(n-k+1-i)$ and $(n-k-i)$
where $i$ varies from $0$ to $n-k-1$. The other neighboring qubits do not have CNOT 
connections due to removal of identity transformations form those qubits. This results
in $n-k$ edges.

\item Now we consider the second kind of connections.
These connections are formed between $l-1$ and $t$ th qubit for 
any $\M^l_t$ transformation. 

There are $n-k$ $SCS^n_K$ blocks with originally $k-1$ $\M$ transformations in $\ck_{n,k}$
which forms the edges:
$$(n-t) \leftrightarrow (n-t-2-i),~~0 \leq i \leq k-2, 0 \leq t \leq n-k-1.$$

Then there are $k-1$ blocks of $SCS^{i+1}_{i}$ with $i-1$ $\M$ transformations
which forms the edges:
$$(k-t) \leftrightarrow (k-t-2-i),~~ 0 \leq i \leq k-t-2, 0 \leq  t \leq k-2.$$

However in $\ckt_{n,k}$ there are no $\M$ transformations of the type 
$\M^{n-k+1+x}_y,~x>0$. Removing such edges $n-k+x \leftrightarrow y$ 
gives us the complete description of the CNOT map of $\ckt_{n,k}$,
which we denote by $G^{n,k}$.
Additionally, in the transformation $\M^{n-k+1}_n$ the edge $n \rightarrow (n-k)$ is not present.  
\end{itemize}
Figure~\ref{fig:cg62} and \ref{fig:cg63} show the CNOT maps $G^{6,2}$ and $G^{6,3}$ respectively. 

\begin{figure}[!ht]
\begin{center}
\tiny
\begin{tikzpicture}[-latex ,auto ,node distance =1.5 cm and 1.5cm ,on grid ,
semithick ,
state/.style ={ circle ,top color =white , bottom color = white ,
draw,black, text=black , minimum width =0.5 cm}]
\node[state] (A) {$1$};
\node[state] (B)[below right=of A] {${2}$};
\node[state] (C) [below=of B] {${3}$};
\node[state] (D) [below left=of C] {$4$};
\node[state] (E) [above left=of D] {$5$};
\node[state] (F) [above=of E] {${6}$};

\path (A) edge (B);\path (B) edge (A);
\path (B) edge (C);\path (C) edge (B);
\path (C) edge (D);\path (D) edge (C);
\path (D) edge (E);\path (E) edge (D); 
\path (D) edge (F);
\path (C) edge (E);\path (E) edge (C);
\path (D) edge (B);\path (B) edge (D);
\path (A) edge (C);\path (C) edge (A);
\end{tikzpicture}
\end{center}
\caption{The CNOT map $G^{6,2}$}
\label{fig:cg62}
\end{figure}
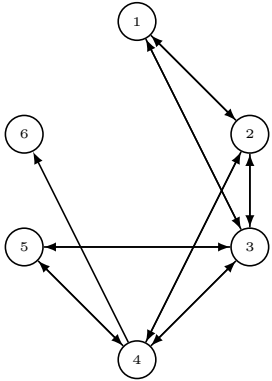

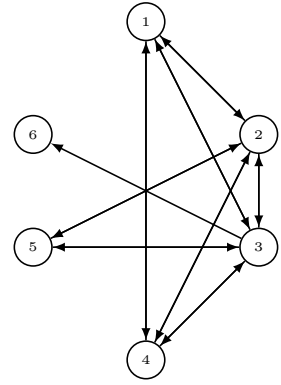
\begin{figure}[!ht]
\begin{center}
\tiny
\begin{tikzpicture}[-latex ,auto ,node distance =1.5 cm and 1.5cm ,on grid ,
semithick ,
state/.style ={ circle ,top color =white , bottom color = white ,
draw,black, text=black , minimum width =0.5 cm}]
\node[state] (A) {$1$};
\node[state] (B)[below right=of A] {${2}$};
\node[state] (C) [below=of B] {${3}$};
\node[state] (D) [below left=of C] {$4$};
\node[state] (E) [above left=of D] {$5$};
\node[state] (F) [above=of E] {${6}$};

\path (A) edge (B);\path (B) edge (A);
\path (B) edge (C);\path (C) edge (B);
\path (C) edge (D);\path (D) edge (C);
\path (C) edge (F);
\path (C) edge (E);\path (E) edge (C);
\path (B) edge (E);\path (E) edge (B);
\path (D) edge (B);\path (B) edge (D);
\path (D) edge (A);\path (A) edge (D);
\path (A) edge (C);\path (C) edge (A);
\end{tikzpicture}
\end{center}
\caption{The CNOT map $G^{6,3}$}
\label{fig:cg63}
\end{figure}

We now count the number of edges in $G^{n,k}$. 
The number of edges present due to $\M$ transformations 
is $nk-\frac{k(k+1)}{2}-n+1-\frac{(k-1)(k-2)}{2}=nk-n-k^2+k$.
There are further $n-k$ edges due to the $\mu$ transformations. 
Which brings the total number of edges in $G^{n,k}$ to $nk-k^2$.
It is important to note that although the number of edges
in $G^{n,k}$ and $G^{n,n-k}$ are same they are not isomorphic. 
Moreover the Graph $G^{n,i}$ is not a subgraph of $G^{n,i+1}$.

Finally we observe that the circuit $\ckt_{n,k}$ can be modified so that
the number of CNOT gates between the qubits change for certain cases, although
the total number of CNOT gates and the overall CNOT map does not change. 
We call these different instances as different CNOT distributions of $\ckt_{n,k}$.

\subsection*{Different CNOT distributions for $\ckt_{n,k}$}

We know from the description of $\ck_{n,k}$~\cite{dicke} that the number of CNOT gates 
in the three qubit transformation $\M$ is reduced from $6$ to $5$ by canceling the last CNOT
of every transformation by rearranging the first two CNOT gates of the next transformation.
Figure~\ref{fig:mod11} shows the original layout as per the algorithm 
and Figure~\ref{fig:mod12} shows the reduction due to \cite{dicke}.

Now let us consider the last transformation ($k$-th) of each $SCS^{n-i}_k,0 < i<n-k$ block, $\M^{n-i-k+1}_{n-i}$.
This transformation acts on the qubits $n-i-k, n-i-k+1$ and $(n-i)$. This is in fact the first transformation
that affects the qubit $(n-i-k)$ and thus the qubit is in the state $\ket{0}$.
If we do not cancel the last CNOT of the preceding $\M$ transformation ($CNOT^{n-i-k+1}_{n-i}$) 
this then enables us to remove of the CNOT gate  ($CNOT^{n-i-k}_{n-i}$), changing the CNOT distribution 
of the circuit without a change in CNOT map or number of CNOT gates. 
This leads to the implementation shown in Figure~\ref{fig:mod13}.
Since there are $n-k-1$ such transformations, this leads to a total of $2^{n-k-1}$ different CNOT distributions.
However, these modifications do not alter the CNOT map of the circuit due to the fact that there are other
CNOT gates applied between these qubits, which is evident from the circuit description.   
As we have observed in Section~\ref{subsec:1} such different CNOT distributions may lead to different number of
CNOT gates acting erroneously on expectation and thus affect the overall error induced in the circuit. 

\begin{figure}
\centering
\begin{minipage}[b]{.4\textwidth}
\hspace{0.5cm}
\Qcircuit @C=1em @R=1em {
&n-k-i&&&& \qw & \qw      &  \ctrl{2} & \targ     & \qw \\
&n-k-i+1&&&& \qw & \ctrl{1} &  \qw      & \ctrl{-1} & \qw \\
&n-i&&&& \qw & \targ    &  \targ    & \qw       & \qw \\
}
\caption{Initial Implementation}
\label{fig:mod11}
\end{minipage}\qquad

\ \\
\begin{minipage}[b]{.4\textwidth}
\hspace{0.5cm}
\Qcircuit @C=1em @R=1em {
& n-k-i &&&& \qw &\targ      &\ctrl{2} & \qw \\
& n-k-i+1&&&& \qw & \ctrl{-1} &\qw      & \qw \\
& n-i  &&&& \qw &\qw        &\targ    & \qw \\
}
\caption{Modification in $\ck_{n,k}$}
\label{fig:mod12}
\end{minipage}
\ \\
\begin{minipage}[b]{.4\textwidth}
\hspace{0.5cm}
\Qcircuit @C=1em @R=1em {
&n-k-i&&&& \qw & \qw      &  \targ     & \qw \\
&n-k-i+1&&&& \qw & \ctrl{1} &  \ctrl{-1} & \qw \\
&n-i&&&& \qw & \targ    &  \qw       & \qw \\
}
\caption{Alternate Implementation}
\label{fig:mod13}
\end{minipage}
\end{figure}

\section{Conclusion}
\label{sec:6}
In this paper we have explored the domain of optimal circuit implementation in terms of
CNOT and single qubit gates.
In this regard we have concisely realized partially defined unitary transformations
to improve the gate count of the most optimal deterministic Dicke state ($\ket{D^n_k}$) 
preparation circuit ($\ck_{n,k}$). 
We have improved the implementation of one such transformation and have also proven the optimality 
of our implementation. We have further improved the Dicke State preparation 
circuit by removing redundant gates and modifying implementations of certain
partially defined unitary transformations 
depending on the active basis states that that act as input to these transformations. 
We have then shown that these improvements not only reduce the number of CNOT and single qubit gates but
also reduces the architectural constraints of the circuit using the case of $\ket{D^4_2}$.
The resultant circuit is the deterministic Dicke State ($\ket{D^n_k},~2 \leq k \leq n-1$) 
Preparation Circuit with the least number of elementary gates to the best of our knowledge.
We have implemented the circuits $\ck_{4,2}$ and the improved circuit $\ckt_{4,2}$ on the
IBM-Q machine ``ibmqx2'' and observed that the deviation from ideal measurement statistics
is significantly lesser in case of $\ckt_{4,2}$.
Furthermore, we have shown that how different CNOT 
distributions can help a circuit without changing the number of gates or the architectural constraints
by comparing the expected CNOT error of two such distributions against a fairly generalized error model. 
We have concluded by describing the CNOT map of the circuit $\ckt_{n,k}$ 
and observe the exponential number of different CNOT distributions that can be derived by modifying the 
circuit to complete our generalization. 

We observed that even the circuits for $\ket{D^5_2}$ could not be implemented 
in the IBM back end machines without adding further CNOT gates to our description.
This is because of incompatibility of the architecture and circuit CNOT maps.  
Therefore it is of all the more importance to form the circuit for an algorithm in the 
most concise way possible. Against this backdrop we have shown how optimally realizing 
partially defined unitary transformations can lead to better implementation results. 
In conclusion we note down the following optimization problems that will help us implement 
algorithms more efficiently in the current scenario. 

\begin{enumerate}

\item Given a maximally partial unitary transformation what is the corresponding unitary matrix 
that can be decomposed using the least number of elementary gates?

\item Given two circuits corresponding to an algorithm 
with isomorphic CNOT maps and the same number of CNOT gates, but
different CNOT distribution across the qubits, 
which circuit will produce less erroneous outcome?

\end{enumerate}

\end{document}